\documentclass[a4paper]{article}  % Comment this line out if you need a4paper
\usepackage[T1]{fontenc}
\usepackage[utf8]{inputenc}
\usepackage[english]{babel}
\usepackage[]{graphics}
\usepackage{wrapfig}
\usepackage{caption}
\usepackage{amsmath}
\usepackage{color}
\usepackage{amssymb}
\usepackage{graphicx}
\usepackage{subfig}
\usepackage{color}
\usepackage[usenames,dvipsnames]{xcolor}
\usepackage{pdflscape}
\usepackage{cite}
\usepackage{graphics} % for pdf, bitmapped graphics files
\usepackage{epsfig} % for postscript graphics files
\usepackage{mathptmx} % assumes new font selection scheme installed
\usepackage{amsmath} % assumes amsmath package installed
\usepackage{amssymb}  % assumes amsmath package installed

\def\build#1_#2^#3{\mathrel{\mathop{\kern0pt#1}\limits_{#2}^{#3}}}%
\def\argmin#1{\build{\rm argmin}_{#1}^{}}

\newtheorem{remark}{Remark}
\newtheorem{assumption}{Assumption}

\newtheorem{theorem}{Theorem}
\newtheorem{lemma}{Lemma}
\newtheorem{proof}{Proof}
\graphicspath{{Figures/}}

\begin{document}

\title{Learning-based predictive control for linear systems:\\a unitary approach} % Title, preferably not more
\date{}
\author{E. Terzi \and L. Fagiano \and M. Farina \and R. Scattolini \\ \textit{e-mail: name.surname@polimi.it}}%\ead{enrico.terzi@polimi.it},    % Add the

\maketitle
%\begin{frontmatter}

\thanks{The authors are with the \textit{Dipartimento di Elettronica, Informazione e Bioingegneria}, Politecnico di Milano, Milano, Italy}
%\address[DEIB]{Dipartimento di Elettronica, Informazione e Bioingegneria, Politecnico di Milano, Milano, Italy}  % Please supply

%%%%%%%%%%%%%%%%%%%%%%%%%%%%%%%%%%%%%%%%%%%%%%%%%%%%%%%%%%%%%%%%%%%%%%%%%%%%%%%%
\begin{abstract}
A comprehensive approach addressing identification and control for learning-based Model Predictive Control (MPC) for linear systems is presented. The design technique yields a data-driven MPC law, based on a dataset collected from the working plant. The method is indirect, i.e. it relies on a model learning phase and a model-based control design one, devised in an integrated manner.\\
\noindent In the model learning phase, a twofold outcome is achieved: first, different optimal $p$-steps ahead prediction models are obtained, to be used in the MPC cost function; secondly, a perturbed state-space model is derived, to be used for robust constraint satisfaction. Resorting to Set Membership techniques, a characterization of the bounded model uncertainties is obtained, which is a key feature for a successful application of the robust control algorithm.\\
\noindent In the control design phase, a robust MPC law is proposed, able to track piece-wise constant reference signals, with guaranteed recursive feasibility and convergence properties. The controller embeds multistep predictors in the cost function, it ensures robust constraints satisfaction thanks to the learnt uncertainty model, and it can deal with possibly unfeasible reference values. The proposed approach is finally tested in a numerical example.
\end{abstract}

%\end{frontmatter}

%%%%%%%%%%%%%%%%%%%%%%%%%%%%%%%%%%%%%%%%%%%%%%%%%%%%%%%%%%%%%%%%%%%%%%%%%%%%%%%%
\section{Introduction} \label{sec:intro}
The idea of combining identification and control for efficient and reliable control systems design, starting from data collected on the plant, has a long standing history, see the survey paper \cite{gevers1996identification}.
%, and \cite{narendra1990identification} for specific classes of nonlinear systems, i.e. neural networks.
Indirect approaches are characterized by an initial phase aimed at estimating the model of the plant, while a following one concerns the model-based control synthesis. In this framework different solutions have been proposed, as thoroughly discussed in \cite{gevers1996identification}. Specifically, in \emph{dual} algorithms, parameter estimation and control design are posed as a combined problem, in \emph{optimal experiment design} methods, identification procedures suitably tailored for the adopted control synthesis algorithm are developed, while in \emph{robust} algorithms the model is estimated together with uncertainty bounds, to be properly used in the control synthesis. With the cheap availability of large data-sets and the advent of more and more powerful identification and learning techniques,
recent years have seen a renaissance of research activity in this area, and in particular on robust methods. From the learning side, new and powerful Set Membership (SM) identification methods, see \cite{CaFS14,LiCM17,TFSM2014,TFFS18a}, have been developed to identify a model for the system with guaranteed prediction error bounds, suitable for robust control design. From the control side, MPC algorithms, robust with respect to model disturbances, have been studied from several standpoints, and considering different characterizations of the model and its associated uncertainty, see e.g.\cite{ChatEtAl:2011,FARINA201653,Schildi:2014,calafiore2013robust,maiworm2015scenario,PoYa93,KoBM96,MAYNE2005219,GoKM06,BeFS13}. Among the most recent contributions in learning-based control, we recall the dual MPC algorithm described in \cite{heirung2017dual} for systems characterized by probabilistic parametric uncertainty and process noise, and the MPC method developed in \cite{aswani2013provably}
guaranteeing both robustness and performance by considering different models of the system. Another recent contribution is reported in \cite{LiCM17}, where an MPC guaranteeing stability has been developed for nonlinear models estimated with the learning method proposed in \cite{calliess2016lazily}.\\ 
In this paper we present a unitary approach to learning-based robust MPC, where we take a joint perspective on the learning and the control design phases. The system generating the data is linear and time-invariant, with unknown order, subject to process disturbance and measurement noise. In the learning phase, we identify with SM different models, together with their uncertainty bounds. Specifically, we compute from data $p$-steps-ahead independent prediction models, $p \in [0,\bar{p}]$, used to compute the future evolution of the system outputs over the prediction horizon $\bar{p}$ considered in the MPC cost function. The use of different models, as previously suggested in \cite{shook1991identification}, allows one to achieve good prediction accuracy at different steps ahead and to have non-conservative bounds on the process disturbance. In addition to the independent $\bar{p}$ models, we also estimate a perturbed state-space model, together with its disturbance bounds, subsequently used in the MPC design for enforcing state, input and output constraints, as well as the robust stability property according to the well known tube-based approach, see \cite{mayne2005robust}.  The robust MPC controller is designed, following the approach proposed in \cite{limon2008automatica}, for tracking piece-wise constant reference signals, with guaranteed recursive feasibility and convergence properties. A  numerical example is finally reported. Preliminary results on the learning and control synthesis algorithms developed in this paper have been reported in \cite{TFFS18a}, and \cite{TFFS18b}. 
In this paper, we propose a novel offline method to learn the uncertainty model to be used in the control design phase, together with a new MPC design to deal with the tracking of (possibly infeasible) piecewise constant reference signals, we derive the full proofs of all the theoretical results concerning learning and control design, and we merge our preliminary work into a unitary and holistic vision of the interplay between learning and control for MPC.\\
%In this paper, we present a unitary approach for learning-based MPC, where we take advantage of a joint perspective to the learning phase and the control design. The system generating the data is linear and time-invariant with measurement noise, but in this case we extend also to the case of process disturbance acting on the system.
%In the proposed approach, we first make use of independent identified prediction models based on Set Membership theory, one for each prediction step, together with the associated bounds computed from data, for which we prove the convergence to the optimal value within the chosen model class.  Then, we rewrite the obtained 1-step model in a state space form with additive state disturbance, whose magnitude has to be quantified, and measurement noise. Such uncertainty magnitude is computed in the least conservative way by means of a linear optimization program.
%Eventually, the resulting state space model, together with  the associated uncertainty, is used in a robust MPC scheme, that explicitly embeds the multistep predictors in the cost function, in order to reduce worst-case error of predictions. We extend the MPC design to the problem of tracking piece-wise constant references and provide theoretical results on recursive feasibility and convergence of the proposed control scheme also in case of unfeasible references. We finally show through numerical examples the effectiveness of the proposed approach.\\
%%
The paper is organized as follows: in Section~\ref{sec:system} the problem is stated and the proposed approach is described. In Section~\ref{sub:identified_multistep} the SM identification algorithm is presented, Section~\ref{sec:control_design} describes the robust MPC control scheme design, followed by a numerical example in Section~\ref{sec:numerical_results} and a concluding discussion in Section~\ref{sec:conclusions}. %All the proofs of the theoretical results are collected in the Appendix.
\section*{Notation}
$k$ is the discrete time index and $\mathbb{Z}$ is the set of non negative integers. The transpose of matrix $M$ is $M^T$. 
%Given two variables $\gamma$ and $\delta$ and positive integers $o$ and $p$, define:
%\begin{eqnarray}
%% \nonumber to remove numbering (before each equation)
%  \Gamma_{o}(k) &=& [\gamma(k),\ldots,\gamma(k-o+1)] \\
%  \Delta_{o}(k) &=& [\delta(k-1),\ldots,\delta(k-o+1)] \\
%  \bar{\Delta}_{p}(k) &=& [\delta(k),\ldots,\delta(k+p-1)]
%\end{eqnarray}
%and the compound vector
%\begin{equation}\label{def_fi}
%  \varphi_{\Gamma,\Delta,\bar{\Delta}}^{(o,p)}(k)=[\Gamma_{o}(k),\Delta_{o}(k),\bar{\Delta}_{p}(k)]^T
%\end{equation}
We denote with $\mathbf{1}_{x}$ a column vector with all its elements equal to one and of dimension $x$, and with $0_{x,y}$ a matrix of zeros with $x$ rows and $y$ columns, wheras $I_x$ denotes the identity matrix of dimension $x$. Finally, $|a|,a \in \mathbb{R}$, denotes the absolute value of real number $a$, $\|v\|=\sqrt{v^Tv}$ denotes the 2-norm of vector $v$, $\|v\|_A=\sqrt{v^TAv}$ denotes the 2-norm of vector $v$ weighted by matrix $A$, and the infinity norm of a generic matrix $N\in \mathbb{R}^{\bar{r}\times \bar{c}}$, with element $n_{rc}$ in position $(r,c)$, is indicated by $\|N\|_{\infty}=\max\limits_{r \in [1\dots \bar{r}]}\sum_{c=1}^{\bar{c}} |{n_{rc}}|$.

\section{Problem formulation: a unitary approach to learning-based MPC}
\label{sec:system}
We consider a discrete-time, linear time-invariant (LTI), single-input/single-output (SISO) system of order $n$ described by the following autoregressive exogenous (ARX) structure ($\cdot^T$ is the matrix transpose operator):
\begin{equation}
\begin{cases}
z(k+1)= \bar{\theta}^{(1)^T} \varphi^{(1)}_z(k) + v(k) \\
y(k)=z(k) + d(k),\\
\end{cases}
\label{eq:realsys}
\end{equation}
where $z$ is the output, $v$ an additive process disturbance, $y$ the output measure, and $d$ an additive measurement noise. For a given integer $p\geq1$, the regressor $\varphi^{(p)}_z(k)\in\mathbb{R}^{2n+p-1}$ is defined as:
\begin{equation}\nonumber
\begin{array}{rcl}
\varphi^{(p)}_z(k)&=&[z(k),\ldots,z(k-n+1),u(k-1),\ldots,u(k-n+1),\\
&&u(k),\ldots,u(k+p-1)]^T,
\end{array}
\end{equation}
\noindent here $u$ is the system input. In \eqref{eq:realsys}, $\bar{\theta}^{(1)}\in\mathbb{R}^{2n+p-1}$ is a vector of unknown system parameters. The value of $n$ is not known a priori as well.

\begin{assumption} (System and signals)\label{ass:du_bounded}\\
- The system \eqref{eq:realsys} is asymptotically stable;\\
- The static gain from $u$ to $z$ is not zero;\\
- $u(k) \in \mathbb{U}\subset\mathbb{R}, \forall k\in \mathbb{Z}$,\; $\mathbb{U}$ compact and convex;\\
- $|d(k)|\leq \bar{d}, \forall k \in \mathbb{Z},\; \bar{d}>0$ known;\\
- $|v(k)|\leq \bar{v}, \forall k \in \mathbb{Z},\; \bar{v}>0$ possibly not known.\hfill$\square$
\end{assumption}
	\begin{remark}\label{R:generality}$\,$\\
		a) The problem is formulated in the SISO setting for the sake of clarity and notational simplicity. \\
		b) Our working assumptions are rather common in theoretical contributions concerned with system identification, when an unknown-but-bounded assumption is considered for process and measurement disturbances. They are valid in many practical applications as well: a characterization of the available sensors can be used to compute the worst-case measurement error bound $\bar{d}$, while for process disturbances we just assume boundedness, without necessarily knowing the worst-case value $\bar{v}$. Indeed, the worst-case effect of the signal $v(k)$ on the system output will be estimated as part of the uncertainty model in our approach.
	\end{remark}
In this paper we adopt an indirect approach to learning-based control synthesis, i.e., based on a sequence of model learning and model-based design phases.\\
The \textbf{learning} phase (Section~\ref{sub:identified_multistep}) has a twofold role:
\begin{itemize}
\item[1.] Identifying optimal (in the sense specified below) independent $p$-steps ahead prediction models of the type
\begin{equation}
\hat{z}(k+p)=\hat{\theta}^{(p)^T}\varphi^{(p)}_y(k)
\label{eq:multistep_models}
\end{equation}
where $\hat{z}(k+p)$ is the predicted output at time $k+p$, and the model regressor $\varphi^{(p)}_y(k)$ is defined as
\begin{equation}\label{eq:phi_y}
\begin{array}{rcl}
\varphi^{(p)}_y(k)&=&[y(k),\ldots,y(k-o+1),u(k-1),\ldots,u(k-o+1),\\
&&u(k),\ldots,u(k+p-1)]^T.
\end{array}
\end{equation}
with $o$ being the order of the prediction model. Models \eqref{eq:multistep_models} are also defined ``multi-step'' since they directly provide the output prediction $p$ steps ahead, without integrating an underlying simulation model. These models, for all $p\in[1,\bar{p}]$, will be used in the MPC cost definition, thanks to their optimal predictive properties, tailored on specific prediction lengths.
\item[2.] Identifying a state-space model of the type
\begin{equation}
\begin{cases}
X(k+1)=AX(k)+B_1u(k)+M_1w(k)\\
z(k)=CX(k)\\
y(k)=z(k)+d(k)
\end{cases}
\label{eq:model_ss}
\end{equation}
where $X$ is the system state, $w$ is the process disturbance, and $A,\,B_1,\,M_1,\,C$ are the system matrices. One of the contributions of this paper consists also of a novel approach for obtaining a non-conservative bound $\bar{w}$ on the amplitude of the process disturbance $w(t)$ from experimental data. This is fundamental, in a constrained robust design context, to limit the conservativeness of the resulting control approach.
\end{itemize}
In the \textbf{control} phase (Section~\ref{sec:control_design}), we propose a scheme, to be applied to the real system \eqref{eq:realsys}, that asymptotically steers the variable $z(k)$ towards the goal $z_{\rm\scriptscriptstyle goal}$ and that guarantees the fulfillment of the following input and output constraints, for all $k\geq 0$.
\begin{subequations}
\label{eq:constraints_YU}
\begin{align}
u(k)&\in\mathbb{U}
\label{eq:constraints_U}\\
z(k)&\in\mathbb{Z}
\label{eq:constraints_Y}
\end{align}
\end{subequations}
where $\mathbb{Z}$ is assumed convex. As already remarked, to this purpose (i) the multi-step prediction models~\eqref{eq:multistep_models} are used for the definition of the cost function and (ii) the perturbed state-space model~\eqref{eq:model_ss}, with bounds $\bar{d}$ and $\bar{w}$ on $d(t)$ and $w(t)$, respectively, are used for constraint satisfaction.
%
%
%
%The problem we consider is the following: exploiting the prior information embodied by the linear structure \eqref{eq:realsys} and Assumption \ref{ass:du_bounded}, together with an available finite data-set of input-output measurements, carry out the following steps:
%1) learn prediction models for the output of system \eqref{eq:realsys} and associated bounds on the worst-case prediction errors;\\
%2) design a MPC algorithm that exploits the models and error bounds learned at step 1) to robustly guarantee closed-loop stability and track an output piece-wise constant reference signal, while robustly satisfying input and output constraints.
%
\section{Learning linear prediction models for robust MPC - a Set Membership approach}\label{sub:identified_multistep}
\subsection{Model structure and preliminary considerations}\label{ss:model_structure}
In MPC with horizon $\bar{p}$, at each step $k$ the predictions of variables $z(k+p)$, $p=1,\ldots,\bar{p}$ are needed. Many contributions on robust MPC in the literature \cite{MRRS00,MAYNE2005219,GoKM06,Mayne14} assume that a model of the system in the form~\eqref{eq:model_ss} is available. A common, but quite conservative, setup is to consider $d(k)=0,\forall k$, and $C=I$, i.e. perfectly measurable state, %(basically consisting of past and current measured - noisy - outputs $y$ and past inputs $u$), 
and finally to assume a known bound $\bar{w}$ on the worst-case additive process disturbance, such that $\|w(k)\|\leq\bar{w},\,\forall k$. Such a model is then integrated forward in time to predict the state and output values at each future step $k+p$.\\
However to learn, from experimental data, a model of the form \eqref{eq:model_ss} with \emph{good prediction accuracy} at different steps ahead and with \emph{non-conservative bounds} on the process disturbance is a complex task. The parameter identification problem is convex only when a 1-step prediction error method is used, which may return models with poor prediction accuracy over multiple future steps (i.e. poor simulation performance), see e.g. \cite{shook1991identification,Farina11}. On the other hand, the use of a cost function that penalizes the multi-step prediction error, or simulation error, yields a nonlinear program (NLP) in the parameters of the 1-step-ahead model, which, besides the possible trapping in local minima, makes it difficult to derive guaranteed disturbance bounds. Additionally, in practical applications the state might not be fully measured, the system order is not known, and  measurement noise and process disturbances are present. These features make the identification problem even more challenging.

To deal with these problems, the approach taken in this paper consists of learning a different (linear-in-the-parameters) multi-step prediction model~\eqref{eq:multistep_models} for each value of $p\in[1,\bar{p}]$. Besides, as discussed, directly using models~\eqref{eq:multistep_models} in the MPC cost function, we will employ their corresponding worst-case error bounds to optimally compute the bound $\bar{w}$ in a state-space realization \eqref{eq:model_ss}, employed to robustly guarantee stability and constraint satisfaction. 

The choice of a model structure of type \eqref{eq:multistep_models} is motivated by the fact that, integrating over time a model of type \eqref{eq:realsys}, the future system outputs are indeed affine in the regressor $\varphi^{(p)}_y(k)$, containing noise-corrupted output measurements:
\begin{equation}
z(k+p)=\bar{\theta}^{(p)^T}\varphi^{(p)}_y(k)+\underbrace{\bar{\theta}_v^{(p)^T}\mathcal{V}^{(p)}(k)}\limits_{\bar{\varepsilon}_v(k)}+\underbrace{\bar{\theta}_d^{(p)^T}\mathcal{D}^{(p)}(k)}\limits_{\bar{\varepsilon}_d(k)}
\label{eq:realsys_integrated}
\end{equation}
where $\mathcal{V}^{(p)}(k)=[v(k),\ldots,v(k+p-1)]$ and $\mathcal{D}^{(p)}(k)=[d(k),\ldots,d(k+p-1)]$ are the sequences of output disturbance and measurement noise, respectively, values from $k$ to $k+p-1$. The parameter vectors $\bar{\theta}^{(p)}$, $\bar{\theta}_v^{(p)}$ and $\bar{\theta}_d^{(p)}$ are polynomial functions of the true system parameters $\bar{\theta}^{(1)}$ (possibly padded with zeros if the model order $o$ is strictly larger than the true system order $n$), readily obtained by recursion of \eqref{eq:realsys}. In our approach, we will consider instead a distinct parameter vector for each $p$, i.e. $\hat{\theta}^{(p)}$ in \eqref{eq:multistep_models}. A first advantage in doing so is the possibility to efficiently compute not only a nominal multi-step prediction model for each $p$, but also a \emph{model set} which is tight (i.e. the smallest one compatibly with the available prior information and data), through a Set Membership (SM) identification approach \cite{Milanese_book}. From such a model set, we can thus estimate a tight worst-case prediction error bound $\tau_p(\hat{\theta}^{(p)})$ for any given multi-step predictor, i.e.:
\begin{equation}
|z(k+p)-\hat{z}(k+p)|\leq\tau_p(\hat{\theta}^{(p)}).
\label{E:uncertainty_bounds}
\end{equation}
We term the bound $\tau_p(\hat{\theta}^{(p)})$ \emph{global}, since it holds for any regressor value $\varphi^{(p)}_y$ within a suitable compact set $\Phi^{(p)}$, introduced in the remainder. A second advantage in using the multi-step models is the possibility to rigorously define, and then efficiently compute, a model optimality criterion and related optimal models, which minimize the worst-case guaranteed prediction error.%For notational simplicity, we present the approach for a generic fixed $p\in[1,\overline{p}]$, the application to any other $p$-value is straightforward.

We describe next the considered data-set, followed by the learning approach. An important assumption throughout the paper is the following.
\begin{assumption} (Model order)\\
	$o \geq n$\hfill$\square$
	\label{ass:modelorder}
\end{assumption}
Assumption \ref{ass:modelorder} can be easily satisfied in practice, on the basis of physical considerations on the system at hand and/or by estimating the system order from data \cite{Ljung99}.
%
%\begin{remark}\label{R:multistep}
%At first glance, it might seem that, by considering the multi-step models, we are neglecting the fact that the future system outputs are determined by a one-step-ahead linear equation iterated over time. As a matter of fact, this information is still captured by the choice of the linear model structure \eqref{eq:multistep_models}, as described above. On the other hand, the multi-step models allow one to avoid the non-convexity of the parameter estimation problem and to derive tight guaranteed accuracy bounds.
%\end{remark}
%
\subsection{Available data-set}
%
%A dataset collected during system operation is assumed to be available, and its features are explained next.
For a given prediction step $p$, we denote with $\Phi^{(p)}$ the compact set  containing all the possible regressor vectors $\varphi_y^{(p)}$, i.e., such that for each $p\in[1,\bar{p}]$
\[
\varphi_y^{(p)}(k)\in\Phi^{(p)}\subset{R}^{2o-1+p},\,\forall k\in\mathbb{Z}.
\]
The set $\Phi^{(p)}$ is not known explicitly in general, as it is a complicated set that depends on the system input and disturbance trajectories and initial conditions of interest. Its compactness is due to the fact that the system is asymptotically stable and its input belongs to a compact set (Assumption \ref{ass:du_bounded}). For any fixed regressor instance $\varphi_y^{(p)}(i)\in\Phi^{(p)}$, considering all possible disturbance sequences $\mathcal{V}^{(p)}(i)$ and all possible noise realizations $\mathcal{D}^{(p)}(i)$, there is a set $\boldsymbol{Y}_p(\varphi_y^{(p)})\subset{R}$ containing all the compatible $p$-step ahead output measurements $y_p(i)\doteq z(i+p)+d(i+p)$. In view of Assumptions \ref{ass:du_bounded} and \ref{ass:modelorder} and of the compactness of $\Phi^{(p)}$, also $\boldsymbol{Y}_p(\varphi_y^{(p)})$ is compact. Let us then define
the set
\begin{equation}\label{E:data_set_continuous}
\mathcal{T}_p=\left\{[{\varphi}_y^{(p)^T},y_p]^T: y_p\in\boldsymbol{Y}_p(\varphi_y^{(p)}), \forall \varphi_y^{(p)}\in\Phi^{(p)}\right\}\subset\mathbb{R}^{2o+p}.
\end{equation}
Now assume that a finite number $N_p$ of data $[\tilde{\varphi}_y^{(p)}(i)^T,\tilde{y}_p{(i)}]^T$
is available, where $\tilde{\varphi}_y^{(p)}(i)$ are the available measured instances of the regressor $\varphi_y^{(p)}(i)\in\Phi^{(p)}$, and  $\tilde{y}_p{(i)}=z(i+p)+d(i+p)$  the corresponding measured values of noise-corrupted outputs. We can express our data-set as:
\begin{equation}\label{E:data_batch2}
\tilde{\mathcal{T}}_p^{N_p}=\left\{[\tilde{\varphi}_y^{(p)}(i)^T,\tilde{y}_p{(i)}]^T,i=1,\dots N_p\right\}\subset\mathbb{R}^{2o+p}.
\end{equation}
The set $\tilde{\mathcal{T}}_p^{N_p}$ is countable and contained in its continuous counterpart $\mathcal{T}_p$. The following assumption is introduced:
\begin{assumption}\label{A:data_set} (Data-set)
	For any $\beta>0$, there exists a value of $N_p<\infty$ such that $d_2\left(\mathcal{T}_p,\tilde{\mathcal{T}}_p^{N_p}\right)\leq\beta$,
	where $d_2\left(\mathcal{T}_p,\tilde{\mathcal{T}}_p^{N_p}\right)\doteq\max\limits_{\tau\in\mathcal{T}_p}\min\limits_{\kappa\in\tilde{\mathcal{T}}_p^{N_p}}\|\tau-\kappa\|_2$ is the Haussdorff distance between the sets $\mathcal{T}_p$ and $\tilde{\mathcal{T}}_p^{N_p}$. \hfill$\square$
\end{assumption}
Assumption \ref{A:data_set} implies that $\lim\limits_{{N_p} \to \infty}d_2\left(\mathcal{T}_p,\tilde{\mathcal{T}}_p^{N_p}\right)=0$, i.e. if more points are added to the data-set, the underlying set of all trajectories of interest will be densely covered. This is essentially an assumption on the persistence of excitation of the inputs used for the preliminary experiments, together with an assumption of bound-exploring property of the additive disturbances $d$ and $v$, such that the bounds $\overline{d}$ and $\overline{v}$ in Assumption \ref{ass:du_bounded} are actually tight.

\subsection{Learning procedure}\label{strategy1}
%\subsubsection{Main steps of the algorithm}\label{preliminaries_SMid}
The proposed estimation procedure consists of the following steps:
\begin{enumerate}
	\item[1.] Define an optimality criterion to evaluate the model estimates and the corresponding optimal (i.e. minimal) error bound.
	\item[2.] Derive a procedure to estimate the optimal error bound.
	\item[3.] Based on the available data and the error bound estimate, build the set of all admissible model parameters (Feasible Parameter Set, FPS).
	\item[4.] Using the information summarized in the FPS, for any given model of the form \eqref{eq:multistep_models} compute the related guaranteed error bound $\tau_p$, see \eqref{E:uncertainty_bounds}.
	\item[5.] Select a nominal model with minimal guaranteed error bound.
\end{enumerate}
\subsubsection{Optimal parameter set and optimal error bound}\label{SS:optim_def}
For any $p\in[1,\bar{p}]$, consider a given value of $\hat{\theta}^{(p)}$. From \eqref{eq:multistep_models} and \eqref{eq:realsys_integrated}, the error between the true system output and the predicted one is, for all $k\in \mathbb{Z}$:
\begin{equation}
\begin{array}{ll}
\epsilon_p(\hat{\theta}^{(p)},\varphi^{(p)}_y(k),\mathcal{V}^{(p)}(k),\mathcal{D}^{(p)}(k))&=z(k+p)-\hat{z}(k+p)\\
&=z(k+p)-\hat{\theta}^{(p)^T}\varphi^{(p)}_y(k)
\end{array}
\label{epsilon_1}
\end{equation}
Thus, from \eqref{eq:realsys} and \eqref{eq:multistep_models} we have:
\begin{equation} \begin{split}
y(k+p)&=z(k+p)+d(k+p)\\
&=\hat{\theta}^{(p)^T}\varphi^{(p)}_y(k)+d(k+p)\\
&+\epsilon_p(\hat{\theta}^{(p)},\varphi^{(p)}_y(k),\mathcal{V}^{(p)}(k),\mathcal{D}^{(p)}(k))
\label{hyp1}
\end{split} \end{equation}
The quantity $\epsilon_p(\cdot,\cdot,\cdot,\cdot)$ accounts for the quality of the estimate $\hat{\theta}^{(p)}$, for possible model order mismatch, and for the disturbances $v$ and $d$.
In view of Assumption \ref{ass:du_bounded}, $\epsilon_p$ is bounded. Moreover, from \eqref{hyp1} :
\begin{equation} \begin{array}{l}
|y(k+p)-\hat{\theta}^{(p)^T}\varphi^{(p)}y(k)|\leq  \\
|\epsilon_p(\hat{\theta}^{(p)},\varphi^{(p)}_y(k),\mathcal{V}^{(p)}(k),\mathcal{D}^{(p)}(k))|+\overline{d}\leq \bar{\epsilon}_p(\hat{\theta}^{(p)})+\overline{d}
\end{array}
\label{E:epsilon_bar}
\end{equation}
where $\bar{\epsilon}_p(\hat{\theta}^{(p)})$ is the global error bound with respect to all possible regressors of interest and all feasible disturbance sequences in the compact set $\Phi^{(p)}$:
\begin{equation}\label{epsilon0}
\begin{array}{rcl}
\bar{\epsilon}_p(\hat{\theta}^{(p)})&=&\min\limits_{\epsilon\in\mathbb{R}}\;\epsilon\\
\text{s.t. } |y_p-\hat{\theta}^{(p)^T}\varphi_y^{(p)}|\leq\epsilon+\overline{d} &&\forall (\varphi_y^{(p)},y_p):\left[\begin{array}{c} \varphi_y^{(p)^T}y_p\end{array} \right]^T\in\mathcal{T}_p
\end{array}
\end{equation}
We can now define the optimal parameter values (i.e. optimal models) as those that minimize the bound $\bar{\epsilon}_p(\hat{\theta}^{(p)})$. As a technical assumption, we consider parameters within a compact set $\Omega^{(p)}\subset\mathbb{R}^{2o+p-1}$. $\Omega^{(p)}$ can take into account application-specific prior information on the model parameters or, if no such information is available, it can be chosen as a large-enough set (e.g. by considering box constraints of $\pm10^{15}$ on each element of the parameter vector). This technical assumption allows us to use maximum and minimum operators instead of supremum and infimum. The set $\bar{\Theta}^{(p)}$ of optimal parameter values is:
\begin{equation}\label{theta^o}
\bar{\Theta}^{(p)}=\left\{\tilde{\bar{\theta}}^{(p)}:\,
\tilde{\bar{\theta}}^{(p)}=\arg\min\limits_{{\theta}^{(p)} \in \Omega^{(p)}}\;\bar{\epsilon}_p(\theta^{(p)})\right\},
\end{equation}
and we denote with $\bar{\epsilon}_p^*$ the corresponding optimal error bound:
\begin{equation}\label{epsilon^o}
\bar{\epsilon}_p^*=\min\limits_{\theta^{(p)} \in \Omega^{(p)}}\;\bar{\epsilon}_p(\tilde{\bar{\theta}}^{(p)}).
\end{equation}
Considering \eqref{epsilon0}-\eqref{epsilon^o} we can alternatively write:
\begin{equation}\label{theta^o2}
\begin{array}{rcl}
\bar{\Theta}^{(p)}&=&\left\{ \tilde{\bar{\theta}}^{(p)}\in \Omega^{(p)}: |y_p-\tilde{\bar{\theta}}^{(p)^T}\varphi_y^{(p)}|\leq\bar{\epsilon}_p^*+\overline{d},\right.\\
										  &&\left.\forall (\varphi_y^{(p)},y_p):\left[ \begin{array}{c} \varphi_y^{(p)^T}y_p\end{array} \right]^T \in\mathcal{T}_p \right\}
\end{array}
\end{equation}
\subsubsection{Estimating the optimal error bound}\label{SS:opt_estim}

The optimal models and optimal error bound cannot be computed in practice, since the solution to \eqref{theta^o} would imply the availability of an infinite number of data and the solution to an infinite-dimensional optimization program. However, we can compute an estimate $\underline{\lambda}_p\approx\bar{\epsilon}_p^{*}$ from the available experimental data, by solving the following linear program (LP):
\begin{equation}
\begin{array}{c}
\underline{\lambda}_p=\min\limits_{\theta^{(p)} \in \Omega^{(p)},\lambda\in \mathbb{R}^+}\,\lambda \\
\text{subject to}\\
|\tilde{y}_p-\theta^{(p)^T}\tilde{\varphi}_y^{(p)}|\leq\lambda+\overline{d},\,\forall (\tilde{\varphi}_y^{(p)},\tilde{y}_p):\left[\begin{array}{c} \tilde{\varphi}_y^{(p)^T}\tilde{y}_p\end{array} \right]^T\in\tilde{\mathcal{T}}_p^{N_p}
\end{array}
\label{E:bound_estim}
\end{equation}
%The condition $\lambda\in \mathbb{R}^+$ is required to enforce a positive estimate of the error bound: without this constraint, the obtained estimate could result to be negative, especially in presence of small amount of data and output disturbance realizations with much smaller magnitude than the considered bound $\overline{d}$.
The following result shows that, under the considered working assumptions, the value of $\underline{\lambda}_p$ converges to the optimal one, $\bar{\epsilon}_p^{*}$.
\begin{theorem}
	Let Assumptions \ref{ass:du_bounded}-\ref{A:data_set} hold. Then:
	\begin{enumerate}
		\item $\underline{\lambda}_p\leq\bar{\epsilon}_p^*$;
		\item $\forall \rho\in(0,\bar{\epsilon}_p^*]\; \exists \; N_p<\infty\; :\; \underline{\lambda}_p\geq\bar{\epsilon}_p^*-\rho$\hfill$\square$
	\end{enumerate}
	\label{convergence}
\end{theorem}
\vspace{0.2cm}
\begin{proof}
See the Appendix.\hfill$\blacksquare$
\end{proof}
Theorem \ref{convergence} implies that $\lim\limits_{N_p\to\infty}(\bar{\epsilon}_p^*-\underline{\lambda}_p)=0^+$, i.e. that the solution of \eqref{E:bound_estim} converges to the optimum from below. This is a consequence of the considered problem settings, where a finite data-set is available. In practice, one can increase the number $N_p$ of experimental data and observe the behavior of $\underline{\lambda}_p$, which converges to a limit provided that the data are informative enough. Then, a practical approach to compensate for the uncertainty caused by the use of a finite number of measurements  is to inflate the value $\underline{\lambda}_p$:
\begin{equation}\label{E:alpha}
\hat{\bar{\epsilon}}_p=\alpha\underline{\lambda}_p,\;\alpha>1.
\end{equation}
With sufficiently large number of and exciting data points, a coefficient $\alpha\simeq1$ can be chosen. We show an example of such a procedure in Section \ref{sec:numerical_results}, and consider the following assumption in the remainder:
\begin{assumption} (Optimal error bound)\label{A:estim_bound}\\
The chosen value of $\alpha$ is such that $\hat{\bar{\epsilon}}_p\geq\bar{\epsilon}_p^*$.\hfill$\square$
\end{assumption}

\subsubsection{Feasible Parameter Set}\label{FeasibleParameterSet}
We exploit the estimated optimal error bound to construct the tightest set of parameter values consistent with all the prior information, i.e. the FPS $\Theta^{(p)}$:
\begin{equation}
\begin{array}{rcl}
\Theta^{(p)}&=&\left\{\theta^{(p)} \in \Omega^{(p)} : |\tilde{y}_p-\theta^{(p)^T}\tilde{\varphi}^{(p)}_y|\leq\hat{\bar{\epsilon}}_p+\overline{d},\right.\\
&&\left.\forall (\tilde{\varphi}^{(p)}_y,\tilde{y}_p):\left[\begin{array}{c} \tilde{\varphi}^{(p)^T}_y\tilde{y}_p\end{array} \right]^T\in\tilde{\mathcal{T}}_p^{N_p}\right\}
\end{array}
\label{FPSdef}
\end{equation}
The set $\Theta^{(p)}$ is non-empty by construction, since under Assumption \ref{A:estim_bound} we have (see \eqref{theta^o2} and \eqref{FPSdef}) that $
\bar{\Theta}^{(p)}\subseteq\Theta^{(p)}$. If the FPS is bounded, it results in a polytope with at most $N_p$ faces. If it is unbounded, then this indicates that the available measured data are not informative enough to derive a bound on the worst-case model error,  and that $N_p$ must be increased until a bounded FPS is obtained. This situation usually occurs when very few data points are used (e.g. $N_p<2o+p-1$) or the preliminary experiments are not informative enough.

\subsubsection{Error bound computation for a given model}\label{SS:error_bounds}
Consider now a given parameter vector $\hat{\theta}^{(p)}$ and any $\varphi^{(p)}_y(k)\in\Phi^{(p)}$. From \eqref{eq:realsys_integrated}, \eqref{epsilon_1}, and
\eqref{theta^o2} it follows that
\begin{equation}
\lvert z(k+p) - \hat{z}(k+p) \rvert
\leq\lvert (\bar{\theta}^{(p)}-\hat{\theta}^{(p)})^T\varphi_y^{(p)}(k)\rvert+\bar{\epsilon}_p^*
\label{prediction error}
\end{equation}
In view of Assumption \ref{A:estim_bound}, the global worst-case prediction error bound for model $\hat{\theta}^{(p)}$ is then:
\begin{equation}
\tau_p(\hat{\theta}^{(p)})=\max\limits_{\varphi_y^{(p)}\in\Phi^{(p)}}\max\limits_{\theta \in\Theta^{(p)}}\lvert (\theta-\hat{\theta}^{(p)})^T \varphi_y^{(p)} \rvert + \hat{\bar{\epsilon}}_p
\label{eq:tau_def}
\end{equation}
This bound cannot be computed exactly with finite data under the considered assumptions, and its computation would be intractable also if the set $\Phi^{(p)}$ were known precisely. The complexity may be reduced only if additional assumptions are made, e.g. that $\Phi^{(p)}$ is a polytope, which however may result in a high conservativeness.
However, we can approximate $\tau_p$ by computing the outer maximization in \eqref{eq:tau_def} over the finite data-set $\tilde{\mathcal{T}}_p^{N_p}$:
\begin{equation}\label{tau_estim}
\underline{\tau}_p(\hat{\theta}^{(p)})=\max\limits_{i=1,\ldots,N_p}\max\limits_{\theta \in\Theta^{(p)}}\lvert (\theta-\hat{\theta}^{(p)})^T\tilde{\varphi}_y^{(p)}(i) \rvert + \hat{\bar{\epsilon}}_p
\end{equation}
The following result shows convergence from below of $\underline{\tau}_p(\hat{\theta}^{(p)})$ to $\tau_p(\hat{\theta}^{(p)})$.
\begin{lemma}
	Let Assumptions \ref{ass:du_bounded}-\ref{A:data_set} hold. Then, for any $\hat{\theta}^{(p)}\in\Omega^{(p)}$:
\begin{enumerate}
	\item $\underline{\tau}_p(\hat{\theta}^{(p)})\leq\tau_p(\hat{\theta}^{(p)})$;
	\item $\forall \rho\in(0,\tau_p(\hat{\theta}^{(p)})]\; \exists \; N_p<\infty\; :\; \underline{\tau}_p(\hat{\theta}^{(p)})\geq\tau_p(\hat{\theta}^{(p)})-\rho$
\end{enumerate}\hfill$\square$
\label{convergence2}
\end{lemma}
\begin{proof}
See the Appendix.\hfill$\blacksquare$
\end{proof}
Considerations similar to those reported after Theorem \ref{convergence} for $\underline{\lambda}_p$ hold also for the bound $\underline{\tau}_p(\hat{\theta}^{(p)})$, i.e. it is possible to monitor its behavior for increasing values of $N_p$ in order to evaluate convergence. As done in \eqref{E:alpha}, we inflate this bound to account for the uncertainty deriving from our finite data-set:
\begin{equation}\label{E:gamma}
\hat{\tau}_p(\hat{\theta}^{(p)})=\gamma\underline{\tau}_p(\hat{\theta}^{(p)}),\;\gamma>1,
\end{equation}
and we assume that the resulting estimate is larger than the true bound:
\begin{assumption} (Error bound for a given $\hat{\theta}^{(p)}$)\label{A:estim_bound2}\\
The chosen value of $\gamma$ is such that $\hat{\tau}_p(\hat{\theta}^{(p)})\geq\tau_p(\hat{\theta}^{(p)})$.\hfill$\square$
\end{assumption}

\subsubsection{Selection of nominal multi-step models}\label{SS:nominal model}

The last step in the proposed estimation algorithm is to select a nominal multi-step model for each prediction step $p$. The most common approach is probably based on least-squares estimation: in this case the results of Section \ref{SS:error_bounds} can be applied to obtain an estimate of the resulting global error bound. Since our final goals are to employ the multi-step models in a robust MPC algorithm and estimate bound $\bar{w}$ for the perturbed model \eqref{eq:model_ss} in a non-conservative manner, we rather seek the model that minimizes the worst-case error bound for each $p$ value. Specifically, considering that the tightest set that contains the optimal parameter values (i.e. with minimum error, see Section \ref{SS:optim_def}) is the FPS $\Theta^{(p)}$, we search within this set for a parameter value that minimizes the resulting bound $\hat{\tau}_p(\hat{\theta}^{(p)})$:
\begin{equation}
\hat{\theta}^{(p)*}=\arg\min\limits_{\hat{\theta}^{(p)}\in\Theta^{(p)}}\hat{\tau}_p(\hat{\theta}^{(p)}).
\label{eq:theta_opt}
\end{equation}
The resulting nominal model reads
\begin{equation}
\hat{z}(k+p)=\hat{\theta}^{(p)*^T}\varphi_y^{(p)}(k)
\label{minmaxmax_orig}
\end{equation}
 and the associated error bound estimate is $\hat{\tau}_p(\hat{\theta}^{(p)*})$.
%\begin{equation}\label{E:nom_bound}
%\begin{array}{c}
%\hat{\tau}_p(\hat{\theta}^{(p)*})=\\
%\gamma\left(\min\limits_{\hat{\theta}^{(p)}\in\Theta^{(p)}}\max\limits_{i=1,\ldots,N_p}\max\limits_{\theta \in\Theta^{(p)}}\lvert (\theta-\hat{\theta}^{(p)})^T%\tilde{\varphi}_y^{(p)}(i) \rvert\right) + \hat{\bar{\epsilon}}_p.
%\end{array}
%\end{equation}
%
Note that term $\hat{\bar{\epsilon}}_p$, see \eqref{tau_estim}, does not depend on $\hat{\theta}^{(p)*}$ and it converges to the optimal error bound $\bar{\epsilon}_p^*$ as $N_p$ increases (Theorem \ref{convergence}).
%Moreover, it can be shown that, under the considered working assumptions, if $\bar{\theta}^{(p)}$ is unique, then the difference $|\hat{\theta}^{(p)*}-\bar{\theta}^{(p)}|$ tends to zero, and the associated error bound \eqref{E:nom_bound} tends to the minimum value $\bar{\epsilon}_p^*$ as well.
%Finally, for any value of $N_p$ it can be shown that $\hat{\tau}_p(\hat{\theta}^{(p)*})$ corresponds to the radius of information, i.e. the minimum guaranteed error value that can be attained with the given prior information and data \cite{Traub80}.
%Finally, the following result confirms the optimality of $\hat{\tau}_p(\hat{\theta}^{(o)*}_p)$ with respect to any 1-step iterated model
%\begin{theorem}
%	Let Assumptions \ref{ass:du_bounded}-\ref{A:data_set} hold. Then, for each $p>1$ the bounds $\hat{\tau}_p(\hat{\theta}^{(o)*}_p)$ are less conservative than the worst-case bounds obtained by iterating any 1-step-ahead prediction model of the form \eqref{eq:1-step-identified-model}.
%	\label{optimality_multistep_1step}
%\end{theorem}
%\vspace{0.2cm}
%\begin{proof}
%See the dedicated paragraph in Section \ref{sec:proofs}
%\end{proof}

\begin{remark} \label{R:computation}  $\hat{\theta}^{(p)*}$ in \eqref{eq:theta_opt} reads
\begin{equation}
\hat{\theta}^{(p)*}=\arg\min\limits_{\hat{\theta}^{(p)}\in\Theta^{(p)}}
\max_{i=1,\ldots,N_p}\max\limits_{\theta \in\Theta^{(p)}}\lvert (\theta-\hat{\theta}^{(p)})^T\tilde{\varphi}_y^{(p)}(i) \rvert.
\label{minmaxmax}
\end{equation}
This problem can be solved by reformulating it as $2\,N_p+1$ LPs \cite{boyd2004convex},\cite{MiTe85}. %The complete reformulation is reported in \cite{TFFS18al}.
\end{remark}

\subsection{Derivation of the state-space model realization and estimation of the corresponding process disturbance bound} \label{sec:CDC_modelforcontrol}
In this section we describe the derivation of the state-space model \eqref{eq:model_ss} and of of the bound $\bar{w}$ of the corresponding disturbance $w(k)$.\\
First of all, we define the equations of the state-space model  \eqref{eq:model_ss} based on the nominal $1$-step ahead predictor, i.e., \eqref{minmaxmax_orig} with $p=1$. To do so, recalling the structure of $\varphi_y^{(p)}$ \eqref{eq:phi_y}, note that we can partition the parameter vector $\hat{\theta}^{(p)}$ of a prediction model as follows:
\begin{equation}
\label{eq:model_partition}
\hat{\theta}^{(p)}=\left[\hat{\theta}^{(p)^T}_{AR}\; \hat{\theta}^{(p)^T}_{U}\; \hat{\theta}^{(p)^T}_{\bar{U}} \right]^T,
\end{equation}
where $\hat{\theta}^{(p)}_{AR}\in\mathbb{R}^o$, $\hat{\theta}^{(p)}_{U}\in\mathbb{R}^{o-1}$ and $\hat{\theta}^{(p)}_{\bar{U}}\in\mathbb{R}^p$ are the parameters pertaining to the past $o$ output values, the past $o-1$ input values, and the current and future inputs, respectively, up to $p-1$ steps ahead.
%\begin{equation}
%\label{eq:model_partition}
%\hat{\theta}^{(1)*}=\left[\hat{\theta}^{(1)*^T}_{AR}\; \hat{\theta}^{(1)*^T}_{U}\; \hat{\theta}^{(1)*^T}_{\bar{U}} \right]^T,
%\end{equation}
%where $\hat{\theta}^{(1)*}_{AR}\in\mathbb{R}^o$ are the parameters pertaining to the past $o$ output values, $\hat{\theta}^{(1)*}_{U}\in\mathbb{R}^{o-1}$ the ones pertaining to the past $o-1$ input values, and $\hat{\theta}^{(1)*}_{\bar{U}}\in\mathbb{R}^p$ is the one pertaining to the current input. 
We define the state vector of the  model \eqref{eq:model_ss} as
\begin{equation}
X(k)=[z(k), \dots , z(k-o+1), u(k-1), \dots,  u(k-o+1) ]^T,
\label{eq:state_z}
\end{equation}
Denoting the process disturbance as $w(k)\in\mathbb{R}$ (accounting for both the disturbance $v$ and prediction error stemming from the learning phase), the state $X(k)$ evolves according to \eqref{eq:model_ss} with the following matrices:
\begin{equation}
\begin{array}{c}
A=\left[\begin{array}{cc}
\hat{\theta}_{AR}^{(1)^*T}&\hat{\theta}^{(1)^*T}_{U}\\
\begin{array}{cc}I_{o-1}&0_{o-1,1}\end{array}&0_{o-1,o-1}\\
0_{o-1,o}&\begin{array}{c}0_{1,o-1}\\\begin{array}{cc}I_{o-2}&0_{o-2,1}\end{array}\end{array}\end{array}\right],\,B_1=\left[\begin{array}{c}
\hat{\theta}_{\bar{U}}^{(1)*}\\		
0_{o-1,1}\\
1\\
0_{o-2,1}\end{array}\right]\\																	
M_1=\begin{bmatrix}1\\
0_{2(o-1),1}\end{bmatrix},\;C=\begin{bmatrix}1 & 0_{1,2o-2}  \end{bmatrix}
\end{array}
\label{eq:matrices}
\end{equation}
\begin{remark}
The model \eqref{eq:matrices} of order $2o-1$ is considered to be in minimal form
\end{remark}
Secondly, we need to define the bound $\bar{w}$ on the amplitude of $w(k)$. 
As anticipated, to this aim we will use the computed FPSs $\Theta^{(p)}$. More specifically, the following approach is proposed.\\
Starting from a noise-corrupted initial state at step $k$ and by iteration of the state-space model \eqref{eq:model_ss} (discarding process disturbance), we can compute a $p$-steps ahead prediction $\hat{z}^{(1)}(k+p)$ of the variable $z(k+p)$ as follows:
\begin{equation}
\label{eq:Id_p_predictor}
\hat{z}^{(1)}(k+p)=CA^pX_y(k)+C\sum_{i=0}^{p-1}A^iB_1u(k+p-i-1)
\end{equation}
where 
\begin{equation}
X_y(k)=[y(k),\ldots,y(k-o+1),\,u(k-1)\ldots,\,u(k-o+1)]^T
\label{eq:state_y}
\end{equation}
We can write \eqref{eq:Id_p_predictor} equivalently as
\begin{equation}
\hat{z}^{(1)}(k+p)=\hat{\theta}^{(p),1^T}\varphi^{(p)}_y(k).
\label{eq:predictor_k_p}
\end{equation}
This is a multi-step prediction model whose parameter vector $\hat{\theta}^{(p),1}$, in view of \eqref{eq:Id_p_predictor}, is composed of polynomial combinations of the entries of the $1$-step ahead prediction model parameter vector $\hat{\theta}^{(1)*}$. Clearly, $\hat{\theta}^{(p),1}$ is in general different from $\hat{\theta}^{(p)*}$ used in \eqref{minmaxmax_orig}, and therefore $\hat{z}(k+p)\neq \hat{z}^{(1)}(k+p)$. At this point, we can use the FPSs derived in Section \ref{FeasibleParameterSet} to estimate the associated worst-case prediction error bounds, $\hat{\tau}_p(\hat{\theta}^{(p),1})$:
\begin{equation}\label{eq:bound_ss_multistep}
|z(k+p)-\hat{z}^{(1)}(k+p) | \leq \hat{\tau}_p(\hat{\theta}^{(p),1})
\end{equation}
On the other hand, by initializing the state-space model with the true (i.e. without measurement noise) initial state \eqref{eq:state_z}, and including the presence of process disturbance $w$, we can also write:
\begin{equation}
\label{eq:Id_p_predictor_system}
\begin{array}{ll}
{z}(k+p)=&CA^pX(k)+C\sum_{i=0}^{p-1}A^i(B_1u(k+p-i-1)\\
&+M_1w(k+p-i-1))
\end{array}
\end{equation}
Then, taking the difference between \eqref{eq:Id_p_predictor_system} and \eqref{eq:Id_p_predictor}, we obtain:
\begin{equation}
\begin{array}{rcl}
z(k+p)-\hat{z}^{(1)}(k+p)&=&C\sum_{i=0}^{p-1}{A^iM_1w(k+p-i-1)}\\
&&-CA^p\left(X_y(k)-X(k)\right),
\end{array}
\label{eq: prediction error}
\end{equation}
which highlights the prediction error due to the process disturbance $w$, and the one due to the measurement noise on the initial condition, $X_y(k)-X(k)$. Note that the latter is equal to zero for all state components pertaining to the past input values, and it is at most equal to $\bar{d}$ for all components pertaining to the past output values. Thus, recalling that $|w(k)|\leq \bar{w}$, we have:
\begin{equation}\label{eq:bound_ss_iteration}
|z(k+p)-\hat{z}^{(1)}(k+p)| \leq  \sum_{i=0}^{p-1}| CA^iM_1 |\bar{w}+\|CA^{p}E\|_{\infty}\bar{d}
\end{equation}
where $E=\left[I_{o}\;0_{(o-1), o}\right]^T$. %Equations \eqref{eq:bound_ss_multistep} and \eqref{eq:bound_ss_iteration} represent two approaches to compute the same quantity, but the disturbance bound $\bar{w}$ in the latter is still to be determined. 
The idea proposed here is to compute $\bar{w}$ as the minimum value such that the bounds \eqref{eq:bound_ss_iteration} do not violate the (tight) bounds \eqref{eq:bound_ss_multistep} for all $p\in[1,\overline{p}]$:
\begin{equation}\label{eq:wcomputation_dknown}
\begin{array}{c}
\bar{w}=\arg\min\limits_{w \in \mathbb{R}^+}w\\
\textit{s.t. } \sum_{i=0}^{p-1}| CA^iM_1 |w+\|CA^{p}E\|_{\infty}\bar{d} \geq \hat{\tau}_p(\hat{\theta}^{(p),1}),\, p\in[1,\bar{p}]
\end{array}
\end{equation}
Note that problem \eqref{eq:wcomputation_dknown} always admits a finite feasible solution thanks to the boundedness of $\hat{\tau}_p(\hat{\theta}^{(p),1}), \forall p\in[1,\bar{p}]$

\section{MPC for tracking with learned models}\label{sec:control_design}
As anticipated in Section~\ref{sec:system}, the MPC controller devised in this paper uses, in the cost function optimized at each time instant $k$, the optimal $p$-steps ahead models~\eqref{eq:multistep_models} to predict in the best possible way the future evolution of the output variable, while the perturbed state-space model~\eqref{eq:model_ss} is used to rigorously define the constraints and ensure recursive feasibility. For notational convenience, we rewrite the predictions \eqref{minmaxmax_orig} as outputs of model \eqref{eq:model_ss} (where matrices $A,B_1,C$ are defined in \eqref{eq:matrices}), as follows:
\begin{align}
z_p(k)=C_p X(k)+D_p U(k)
\label{eq:perturbedy}
\end{align}
where $U(k)=\begin{bmatrix}u(k) & \dots& u(k+\bar{p})\end{bmatrix}$
$C_p=\begin{bmatrix}\hat{\theta}^{(p)*^T}_{{AR}} & \hat{\theta}^{(p)*^T}_{U}\end{bmatrix}$, $D_p=\begin{bmatrix}
\hat{\theta}^{(p)*^T}_{\bar{U}} & 0_{1,\overline{p}+1-p}\end{bmatrix}$ and we denote $z_p(k)=\hat{z}(k+p)$ for brevity.
For later use we also define $C_0=C$ and $D_0=0_{1,\overline{p}+1}$ such that we can write $z(k)=z_0(k)=C_0X(k)+D_0U(k)$.
\subsection{State observer and tube-based control approach}
\label{sec:sub:Observer}
Since $z(k)$ is measured with some noise, the state $X(k)$ of the system \eqref{eq:model_ss} cannot be perfectly reconstructed as a suitable collection of the past available outputs and inputs. For this reason, a Luenberger state observer is employed. To design the observer on the basis of the model \eqref{eq:model_ss}, it is beneficial to introduce an estimate $\hat{w}(k)$ of the disturbance $w$. The term $\hat{w}(k)$ will result from a suitable optimization problem introduced later on, in Section \ref{subsec:cost_fcn}. The observer takes then the following form:
\begin{equation}
\hat{X}(k+1)=A\hat{X}(k)+B_1 u(k)+M_1\hat{w}(k)+L(y(k)-C\hat{X}(k))
\label{eq:state_obs}
\end{equation}
where $\hat{X}(k)$ is the estimated state and the matrix $L$ is chosen such that the closed-loop matrix $(A-LC)$ is Schur stable.\\
Furthermore, for application of a tube-based robust control method inspired by~\cite{mayne2005robust}, we define the nominal dynamic system related to \eqref{eq:model_ss}, where again the disturbance estimate $\hat{w}(k)$ is included, i.e.
\begin{equation}
\bar{X}(k+1)=A\bar{X}(k)+B_1\bar{u}(k)+ M_1 \hat{w}(k)
\label{eq:nomstate_eq}
\end{equation}
The input $u(k)$, to be applied to system \eqref{eq:model_ss} at time instant $k$, is defined as the sum of two components as follows.
\begin{equation}
u(k)=\bar{u}(k)+K(\hat{X}(k)-\bar{X}(k))
\label{eq:controllaw}
\end{equation}
The second component (i.e., $K(\hat{X}(k)-\bar{X}(k))$) is given by a suitable proportional control law, aiming to reduce the displacement of the state $\bar{X}(k)$ of \eqref{eq:nomstate_eq} with respect to the state estimate $\hat{X}(k)$, available at time $k$. The gain $K$ is defined in such a way that the closed-loop transition matrix $A+B_1K$ is Schur stable, e.g. by pole-placement or LQR design.
The corresponding nominal outputs are, for all $p \in [0,\bar{p}]$
\begin{align}
\bar{z}_p(k)=C_p \bar{X}(k)+D_p \bar{U}(k)
\label{eq:perturbedy_nominal}
\end{align}
where $\bar{U}(k)=\begin{bmatrix}\bar{u}(k) & \dots& \bar{u}(k+\bar{p})\end{bmatrix}^T$. We finally define $\bar{z}(k)=C_0\bar{X}(k)=\bar{z}_0(k)$.
\subsection{Definition of the cost function}\label{subsec:cost_fcn}
The goal is to steer variable $z(k)$ in order to track the (possibly piece-wise) constant goal $z_{\rm\scriptscriptstyle goal}$. However, tracking this value could lead to infeasibility problems: to avoid them, inspired by \cite{limon2008automatica}, we introduce an output reference $z_{\rm\scriptscriptstyle ref}$ to be used as a further degree of freedom in the optimization problem.\\
Assuming that a reliable estimate $\hat{\mu}$ of the system gain is available (see the following Remark \ref{rem4}), we can now compute $\hat{w}$ as a function of a generic $z_{\rm\scriptscriptstyle ref}(k)$ as follows. We first compute the constant input and state values $u_{\rm\scriptscriptstyle ref}$ and $X_{\rm\scriptscriptstyle ref}$,  corresponding to the reference output $z_{\rm\scriptscriptstyle ref}$:
\begin{equation}
u_{\rm \scriptscriptstyle ref}(k)=( \hat{\mu}^{} )^{-1} z_{\rm \scriptscriptstyle ref}(k), \quad X_{\rm \scriptscriptstyle ref}(k)=Nz_{\rm \scriptscriptstyle ref}(k)
\label{eq:Xref_uref}
\end{equation}
where $N=\begin{bmatrix}
\mathbf{1}_{o}\\
\mathbf{1}_{o-1}(\hat{\mu}^{})^{-1}
\end{bmatrix}$. The value of $\hat{w}$ can now be defined in such a way that 
\begin{equation}
X_{\rm\scriptscriptstyle ref}(k)=AX_{\rm\scriptscriptstyle ref}(k)+B_1u_{\rm\scriptscriptstyle ref}(k)+M_1\hat{w}(k)
\label{eq:Xref_ss}
\end{equation} i.e., as a linear function of $z_{\rm\scriptscriptstyle ref}$. In short we write
\begin{equation}
\hat{w}(k)= \eta_{zw} z_{\rm\scriptscriptstyle ref}(k)
\label{eq:what_zref}
\end{equation}
where $\eta_{zw}=M_1^T \left[ (I_{2o-1}-A)N - B_1 (\hat{\mu}^{})^{-1} \right]$. Moreover, for consistency, the term $\hat{w}(k)$ is forced to be bounded, so that $|\hat{w}(k)| \leq \bar{w}$.

\begin{remark}
Since the long-term prediction capabilities are commonly more accurate with model \eqref{minmaxmax_orig} with the longest possible prediction horizon, i.e., $p=\bar{p}$, one suitable option for the gain estimate $\hat{\mu}$ is to choose $\hat{\mu}=\mu^{\bar{p}}$, where $$\mu^{\bar{p}}=\frac{(\hat{\theta}_{\bar{U}}^{(\bar{p})*^T}\mathbf{1}_{\bar{p}}+\hat{\theta}^{(\bar{p})*^T}_{U}\mathbf{1}_{o-1} )}{1- \hat{\theta}^{(\bar{p})*^T}_{AR}  \mathbf{1}_{o}}$$ is the gain of the optimal $\bar{p}$-steps-ahead model .% In the following we will consider $\hat{\mu}=\mu^{\bar{p}}$.
\label{rem4}
\end{remark}
Last we can define, $\forall p \in [1,\bar{p}]$, the reference for the $p$-steps ahead model, i.e.,
\begin{equation}
z_{\rm\scriptscriptstyle ref}^{p}(k)=\begin{bmatrix}C_p&D_p\end{bmatrix}
\begin{bmatrix}X_{\rm\scriptscriptstyle ref}(k)\\u_{\rm\scriptscriptstyle ref}(k)\mathbf{1}_{\bar{p}+1}\end{bmatrix}
\label{eq:zrefp}
\end{equation}
and $z_{\rm\scriptscriptstyle ref}^0(k)=z_{\rm\scriptscriptstyle ref}(k)=C_0X_{\rm\scriptscriptstyle ref}(k)+D_0u_{\rm\scriptscriptstyle ref}(k)\mathbf{1}_{\bar{p}+1}$.

The cost function to be minimized at each (sampling) time $k$ is therefore
\begin{equation}\label{eq:cost_fcn}
\begin{split}
J(k)= &\sum_{p=0}^{\bar{p}} \left( \| \bar{z}_p(k) - z_{\rm\scriptscriptstyle ref}^{p}(k) \|_{Q_p}^2 + \|\bar{u}(k+p)-u_{\rm\scriptscriptstyle ref}(k) \|_{R_p}^2 \right) \\
		&+\|\bar{X}(k+\bar{p}+1)- X_{\rm\scriptscriptstyle ref}(k)\|^2_P + \sigma \|z_{\rm\scriptscriptstyle ref}(k) - z_{\rm\scriptscriptstyle goal} \|^2
\end{split}
\end{equation}
where $\bar{X}(k+\bar{p}+1)$ is obtained by iterating the unperturbed state equation \eqref{eq:nomstate_eq} $\bar{p}+1$ times, i.e.,
\begin{equation} \bar{X}(k+\bar{p}+1)=A^{\bar{p}+1}\bar{X}(k)+\Gamma \bar{U}(k) + \Gamma_w \mathbf{1}_{\bar{p}+1}\hat{w}(k)
\label{Xhat}
\end{equation}
with $ \quad\Gamma=\begin{bmatrix}A^{\bar{p}}B_1&\dots&B_1\end{bmatrix}, \Gamma_w=\begin{bmatrix}A^{\bar{p}}M_1&\dots&M_1\end{bmatrix} $.
To compute the weights $Q_p$, $R_p$, and $P$ in order to guarantee closed-loop stability, we must first define $B=\begin{bmatrix}B_1&0_{2o-1,\bar{p}}\end{bmatrix}$ and
\begin{align*}
&\Psi=\begin{bmatrix}CA&CB+DH_1\\
\vdots&\vdots\\
C_{\bar{p}}A&C_{\bar{p}}B+D_{\bar{p}}H_1
\end{bmatrix},\,\bar{\Psi}=\begin{bmatrix}C_1&D_1\\
\vdots&\vdots\\
C_{\bar{p}}&D_{\bar{p}}\\
A^{\bar{p}+1}&\Gamma\\
0_{\bar{p}+1,2o-1}&I_{\bar{p}+1}\end{bmatrix},																	 \notag \\
&H_1=\begin{bmatrix}0_{\bar{p},1}&I_{\bar{p}}\\
0&0_{1,\bar{p}}\end{bmatrix}
\end{align*}
Also, we write $\mathcal{Q}=$diag$(Q_0,\dots,Q_{\bar{p}})$, and \break
$\bar{\mathcal{Q}}=$diag$(Q_1,\dots,Q_{\bar{p}},T_N,\mathcal{R})$, where $T_N$ is a positive definite matrix to be used as a further tuning knob and $\mathcal{R}=$diag$(R_0/2, R_1-R_0,\dots, R_{\bar{p}}-R_{\bar{p}-1})$. Then, the weighting matrices are computed such that the following constraints are satisfied:
\begin{subequations}
\label{eq:LMI_TOT}
\begin{align}
&(A+B_1K)^TP(A+B_1K)-P=-T_N-K^TR_{\overline{p}}K\label{eq:LMI1_p}\\
&\Psi^T\mathcal{Q}\Psi \leq \bar{\Psi}^T\bar{\mathcal{Q}}\bar{\Psi}\label{eq:LMI2_p}\\
&T_N>0, \quad \mathcal{R}>0, \quad \mathcal{Q}>0, \quad P>0 \label{eq:Posdef}
\end{align}
\end{subequations}
Finally,  the scalar $\sigma>0$ \eqref{eq:cost_fcn} must be chosen sufficiently large to provide converge properties, its quantitative evaluation is discussed in the Appendix.
\subsection{Definition of the tightened constraints}
As common in tube-based control, see e.g., \cite{mayne2005robust,Mayne:2009}, we enforce the input and output constraints \eqref{eq:constraints_YU} with suitable tightened bounds on the nominal input and output $\bar{u}(k)$ and $\bar{z}(k)=C\bar{X}(k)$ respectively. For their definition, we first have to define the state estimation error $\hat{e}(k)=X(k)-\hat{X}(k)$. We obtain, from \eqref{eq:model_ss} and~\eqref{eq:state_obs}, that
\begin{equation}
\hat{e}(k+1)=(A-LC)\hat{e}(k)+M_1(w(k)-\hat{w}(k) ) - Ld(k)
\label{eq:errordyn}
\end{equation}
Denote now with $\hat{\mathbb{E}}$ a robust positively invariant (RPI)  set \cite{RakovicKouramas2007} (minimal, if possible) for the system \eqref{eq:errordyn} containing $\hat{e}(0)=X(0)-\hat{X}(0)$, where $|w(k)-\hat{w}(k)| \leq 2\bar{w}$ . This guarantees that, for all $k\geq 0$, $\hat{e}(k)\in\hat{\mathbb{E}}$.\\
We also define the displacement between the estimated state and the nominal one as $\bar{e}(k)=\hat{X}(k)-\bar{X}(k)$. From~\eqref{eq:state_obs} and \eqref{eq:nomstate_eq} we derive 
\begin{equation}
\bar{e}(k+1)=(A+B_1K)\bar{e}(k)+LC\hat{e}(k)+Ld(k)
\label{eq:errordynbar}
\end{equation}
Since the equivalent disturbance $LC\hat{e}(k)+Ld(k)$ is bounded for all $k\geq 0$, we can define as $\bar{\mathbb{E}}$ the (minimal, if possible) RPI set for \eqref{eq:errordynbar}.\\
Then, the input and output constraints can be defined  with reference to the model \eqref{eq:nomstate_eq} in a tightened fashion:%, i.e., to enforce~\eqref{eq:constraints_YU} it will be sufficient to enforce the following.
\begin{equation}
\bar{u}(k)\in\bar{\mathbb{U}}, \quad
\bar{z}(k)\in\bar{\mathbb{Z}}, \quad
\hat{w}(k) \in \mathbb{W},
\label{eq:tight_constr}
\end{equation}
where $\mathbb{W}=\{w \in \mathbb{R}: |w|\leq \bar{w}\}$ and the sets $\bar{\mathbb{U}}$ and $\bar{\mathbb{Z}}$ are closed and satisfy:
\begin{subequations}
\begin{align}
\bar{\mathbb{U}}&\subseteq \mathbb{U}\ominus K\bar{\mathbb{E}}\\
\bar{\mathbb{Z}}&\subseteq \mathbb{Z}\ominus C(\bar{\mathbb{E}}\oplus \hat{\mathbb{E}}) \label{eq:constrZ}
\end{align}
\end{subequations}
Finally, to define the terminal constraint set we consider the following auxiliary control law
\begin{equation}\bar{u}(k)=u_{\rm\scriptscriptstyle ref}(k)+K(\bar{X}(k)-X_{\rm\scriptscriptstyle ref}(k))
\label{eq:nominal aux_law}
\end{equation}
To compute an invariant set where $(\bar{X}(k),z_{\rm\scriptscriptstyle ref})$ must lie in order to guarantee that constraints~\eqref{eq:tight_constr} are verified for all $k$,  we need to define the Maximal Output Admissible Set (MOAS, see \cite{gilbert1991linear}) $\mathbb{O}$ for the  system
\begin{align}
\begin{bmatrix}
\bar{X}(k+1)\\z_{\rm\scriptscriptstyle ref}(k+1)
\end{bmatrix}&=\underbrace{\begin{bmatrix}
	A+B_1K & B_1M_2+M_1\eta_{zw}\\
	0_{1,2o-1} & 1\end{bmatrix}}_{\mathcal{F}}\begin{bmatrix}
\bar{X}(k)\\z_{\rm\scriptscriptstyle ref}(k)
\end{bmatrix}\label{eq:aux_nominal system01_1}
\end{align}
that is subject to the auxiliary control law \eqref{eq:nominal aux_law}, where
$M_2=\hat{\mu}^{-1}-KN$. The triplet $(\bar{u}(k),\bar{z}(k),\hat{w}(k))$ is computed as
\begin{align}
\begin{bmatrix}
\bar{z}(k)\\\bar{u}(k) \\ \hat{w}(k)
\end{bmatrix}&=\underbrace{\begin{bmatrix}
	C&0\\
	K&M_2\\
	0_{1,2o-1} & \eta_{zw} \end{bmatrix}}_{\mathcal{C}}\begin{bmatrix}
\bar{X}(k)\\z_{\rm\scriptscriptstyle ref}(k)
\end{bmatrix}\label{eq:aux_nominal system01_out}
\end{align}
%
%
%Consider the following
%\begin{assumption}\hfill
%\begin{itemize}
%\item[i)] the pair $(\mathcal{F},\mathcal{C})$  is observable,
%\item[ii)] $\bar{\mathbb{X}}_{\mathbb{Z}\mathbb{U}\mathbb{W}}=\bar{\mathbb{Z}}\times\bar{\mathbb{U}} \times \mathbb{W}$ is a closed polytope,
%\end{itemize}
%\end{assumption}
%\hfill$\blacksquare$\\
An invariant, polytopic inner approximation $\mathbb{O}_{\epsilon}$ to the MOAS can be computed in a finite number of steps as shown in \cite{gilbert1991linear}. Specifically,  $\mathbb{O}_{\epsilon}$ is defined as follows
\begin{align}
\mathbb{O}_{\epsilon}=\{(\bar{X},z_{\rm\scriptscriptstyle ref}):\mathcal{C}\mathcal{F}^k(\bar{X},z_{\rm\scriptscriptstyle ref})\in\bar{\mathbb{X}}_{\mathbb{Z}\mathbb{U}\mathbb{W}}\text{ for all }k\geq0
\nonumber\\\text{ and }\lim_{k\rightarrow +\infty}\mathcal{C}\mathcal{F}^k(\bar{X},z_{\rm\scriptscriptstyle ref})\in \bar{\mathbb{X}}_{\mathbb{Z}\mathbb{U}\mathbb{W}}(\epsilon)\}
\label{eq:O_def}
\end{align}
where $\bar{\mathbb{X}}_{\mathbb{Z}\mathbb{U}\mathbb{W}}(\epsilon)$ is a close and compact set satisfying $\bar{\mathbb{X}}_{\mathbb{Z}\mathbb{U}\mathbb{W}}(\epsilon) \oplus \mathcal{B}_{\epsilon}^{3}(0)\subseteq \bar{\mathbb{X}}_{\mathbb{Z}\mathbb{U}\mathbb{W}}$,  with $\mathcal{B}_{\epsilon}^3(0)$  a ball in $\mathbb{R}^3$ containing the origin and with radius $\epsilon$ arbitrarily small.
Note that, see again \cite{gilbert1991linear}, $\mathbb{O}_{\epsilon} \subset \mathbb{O}$ and if $(\mathcal{F},\mathcal{C})$ is observable, then $\mathbb{O}$ is bounded.\\
\subsection{The optimization problem and main result}
The optimization problem, to be solved at each time instant $k\geq 0$, reads
\begin{subequations}\label{eq:optprb}
\begin{align}
J(k|k)=\min_{\bar{X}(k),\bar{U}(k),z_{\rm\scriptscriptstyle ref}(k)}J(k)
\end{align}
subject to the dynamical system~\eqref{eq:nomstate_eq}, \eqref{eq:Xref_uref}, \eqref{eq:Xref_ss}, \eqref{eq:zrefp} and
\begin{align}
\label{eq:constraint_optE}
\hat{X}(k)-\bar{X}(k)\in\bar{\mathbb{E}}
\end{align}
Also, $\forall p \in [0,\bar{p}]$
\begin{equation}
\bar{u}(k+p)\in\bar{\mathbb{U}}, \quad \label{eq:constraint_opt}
\bar{z}(k+p)\in\bar{\mathbb{Z}}, \quad
\hat{w}(k) \in \mathbb{W}
\end{equation}
Finally, as a terminal constraint, the following must be fulfilled
\begin{align}
\begin{bmatrix}
\bar{X}(k+\bar{p}+1)\\z_{\rm\scriptscriptstyle ref}(k)
\end{bmatrix}\in \mathbb{O}_{\epsilon}\label{eq:constraint_optF}
\end{align}
\end{subequations}
If available, the solution to the optimization problem \eqref{eq:optprb} is denoted
$\bar{X}(k|k), \bar{U}(k|k)=(\bar{u}(k|k),\dots,\bar{u}(k+\overline{p}|k)),z_{\rm\scriptscriptstyle ref}(k|k)$, and $u(k)$ in \eqref{eq:controllaw} is applied to the system according to the receding horizon principle. Also, we denote with $\bar{X}(k+p|k)$ the future nominal state predictions generated using \eqref{eq:nomstate_eq} with input $\bar{U}(k|k)$. The following result holds.
\begin{theorem}
\label{thm:res1}
If the optimization problem is feasible at time step $k=0$ then it is feasible at all time steps $k>0$ and, for all $k\geq 0$, the constraints \eqref{eq:constraints_YU} are satisfied. Also, if $\sigma$ is sufficiently large, the resulting MPC control law asymptotically steers the nominal system output $\bar{z}(k)$ to the admissible set-point $z_{\rm\scriptscriptstyle goal}^{{\rm\scriptscriptstyle FEASIBLE}}$, where
\footnotesize
\begin{equation}
\begin{array}{cc}
z_{\rm\scriptscriptstyle goal}^{{\rm\scriptscriptstyle FEASIBLE}} &= \argmin {\tilde{z}} \|\tilde{z}-z_{\rm\scriptscriptstyle goal}\|^2\\
& s.t.\begin{bmatrix}\tilde{z}&\tilde{u}&\tilde{w}\end{bmatrix}^T=\begin{bmatrix} CN & \hat{\mu}^{-1} & \eta_{zw}\end{bmatrix}^T \tilde{z}\in \bar{\mathbb{X}}_{\mathbb{Z}\mathbb{U}\mathbb{W}}(\varepsilon)
\end{array}
\label{eq:yhathat}
\end{equation}
\normalsize
Finally, $\text{dist}(z(k),z_{\rm\scriptscriptstyle goal}^{{\rm\scriptscriptstyle FEASIBLE}}\oplus C(\bar{\mathbb{E}}\oplus\hat{\mathbb{E}})) \to 0$ as $k \to \infty$, where $\text{dist}(\alpha,\beta)$ denotes the distance from point $\alpha$ to set $\beta$.\hfill{}$\square$
\end{theorem}
\begin{proof}
	See the Appendix.\hfill$\blacksquare$
\end{proof}

\section{Simulation example}\label{sec:numerical_results}

The proposed approach for learning-based predictive control has been tested on a simulation example. The considered system is of third order, and it corresponds to the discretization of the system with continuous time transfer function
\begin{equation}
\frac{Z(s)}{U(s)}=G(s)=\frac{160}{(s+10)(s^2+1.6s+16)}
\end{equation}
characterized by dominant complex poles with natural frequency $\omega_n=4$ and damping factor $\xi=0.2$, and with unitary gain. Figure \ref{Figure_openloop} shows the open loop step response of the system under analysis.
The input and output samples are collected with sampling time $T_s=0.1$, the output $z(k)$ is corrupted by an additive disturbance $v(k)$ such that $|v(k)| \leq \bar{v}=0.01$, while the bound on the measurement noise is $\bar{d}=0.1$.
The multistep models and bounds are computed up to $\bar{p}=20$ steps ahead, while the chosen model order is $o=4$.
The collected dataset is composed overall of 1000 input-output data samples, where the input is a step-wise sequence taking a random value in $\{-1,0,1\}$ every $5$ time units.

\begin{figure}
\centering
\includegraphics[scale=0.39]{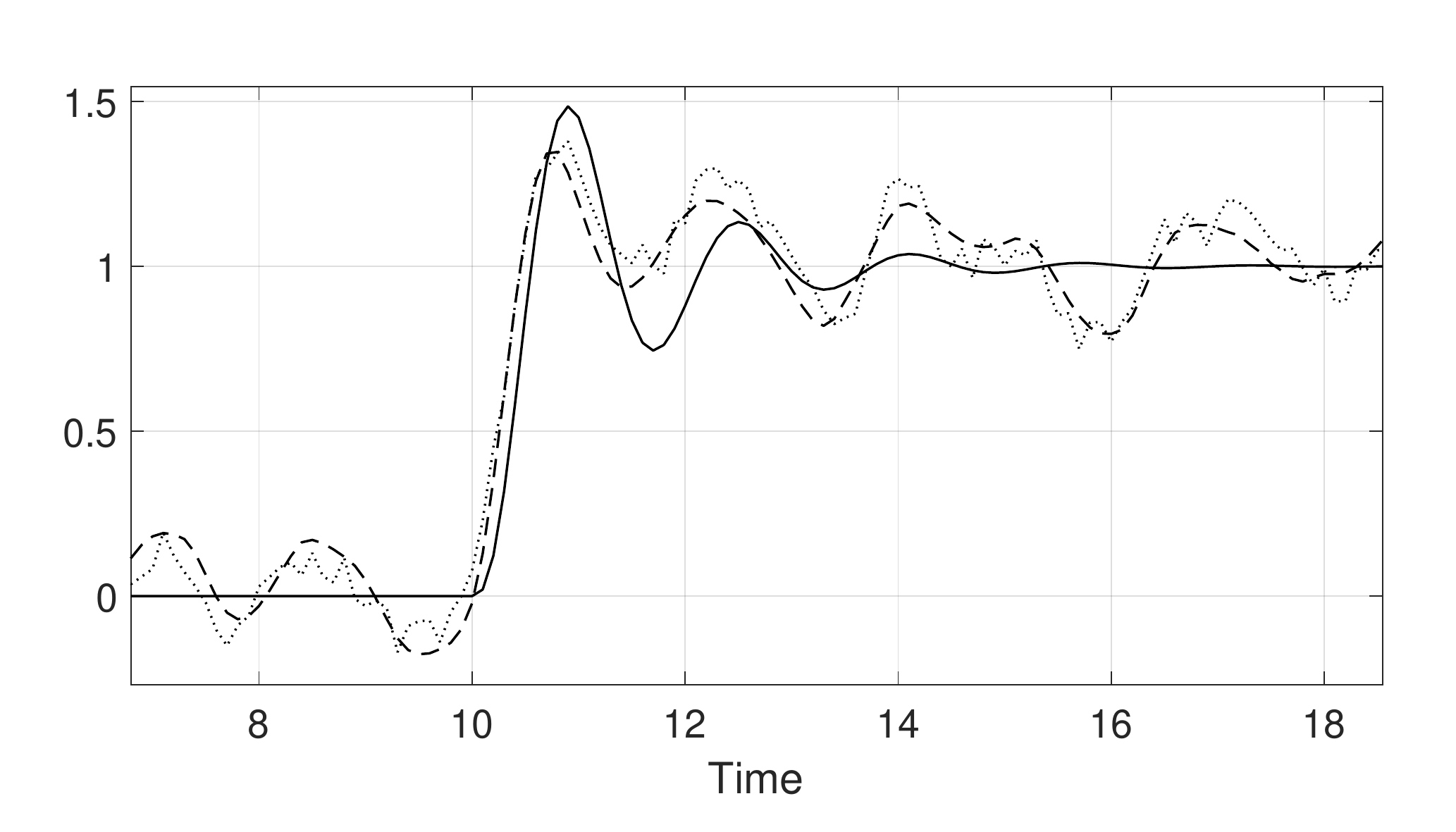}
\caption{Open loop response of the system to a unitary step at time 10. Solid line: nominal $z$ $(v(k)=d(k)=0)$, dashed line: output $z$, dotted line: output measure $y$ }
\label{Figure_openloop}
\end{figure}

Following the approach of Section \ref{SS:opt_estim}, we compute $\underline{\lambda}_p$ and monitor its trend against the percentage of dataset used to compute it. This procedure enables one to assess the convergence rate of $\underline{\lambda}_p$ to a limit value, presumably equal to $\bar{\epsilon}^*_p$ according to Theorem \ref{convergence}. Figure \ref{Figure:Lambda_dataset} shows this trend.
\begin{figure}
\centering
\includegraphics[scale=0.38]{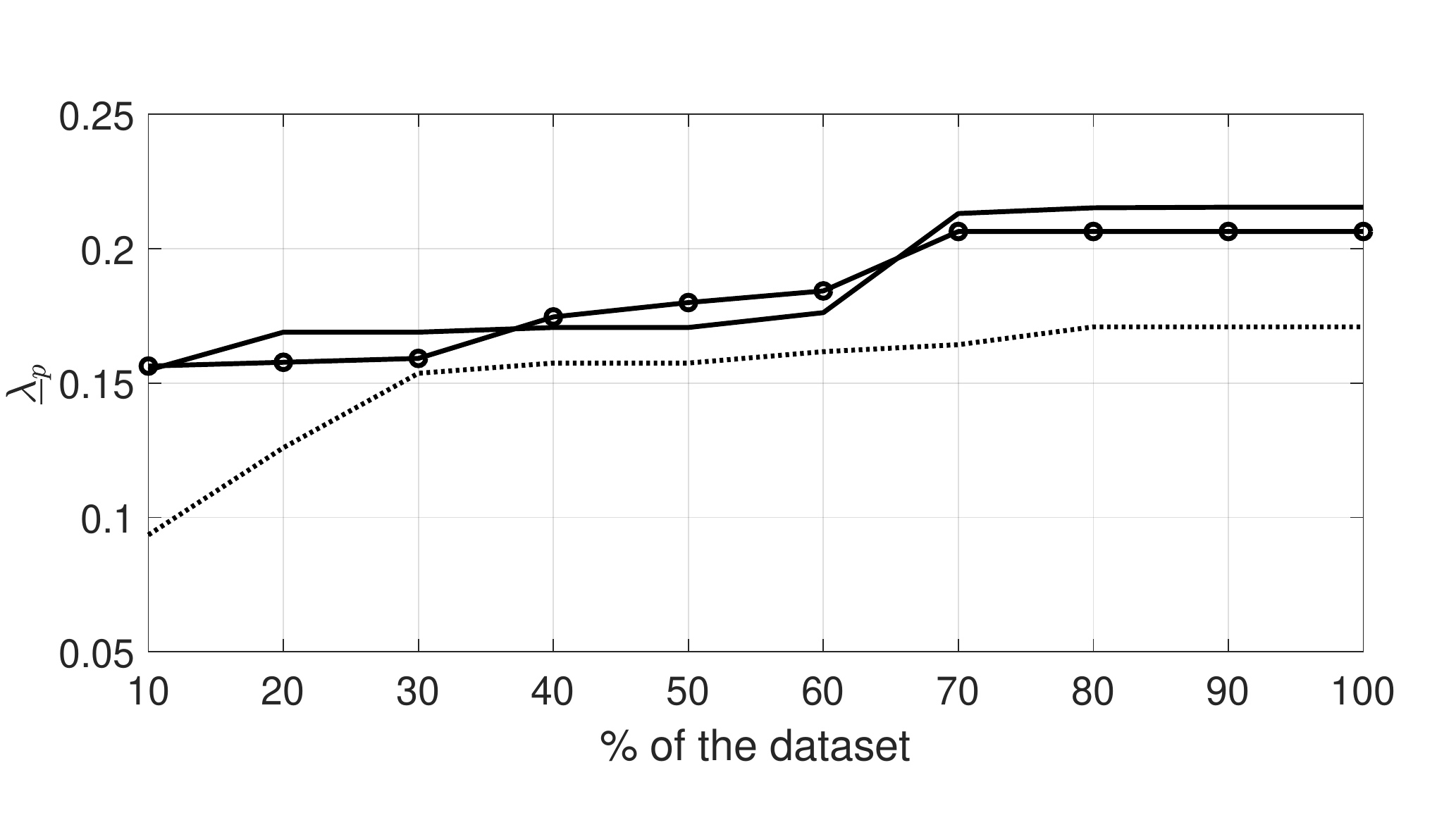}
\caption{Trend of $\underline{\lambda}_p$ against the employed percentage of the dataset. Dotted line: $p$=3, solid line: $p$=10, line with circles: $p=\bar{p}=$20 }
\label{Figure:Lambda_dataset}
\end{figure}
Parameters $\hat{\bar{\epsilon}}_p, p \in [1,\bar{p}]$ are then computed with a conservative factor $\alpha=1.1$ to account for the finite dataset employed, see \eqref{E:alpha}, and the FPSs are then built independently for each step as in \eqref{FPSdef}.
The parameters of the nominal one-step predictor $\hat{\theta}^{(1)*}$ are computed by solving \eqref{eq:theta_opt}.
In order to learn the uncertainty bound $\bar{w}$, the given predictor $\hat{\theta}^{(1)*}$ is iterated and rewritten in form \eqref{eq:predictor_k_p} and the estimated value $\hat{\tau}_p(\hat{\theta}^{(p),1})$ is computed exploiting the FPSs previously introduced. Finally the optimization program \eqref{eq:wcomputation_dknown} is solved leading to the minimizer $\bar{w}$.
\begin{figure}
\centering
\includegraphics[scale=0.39]{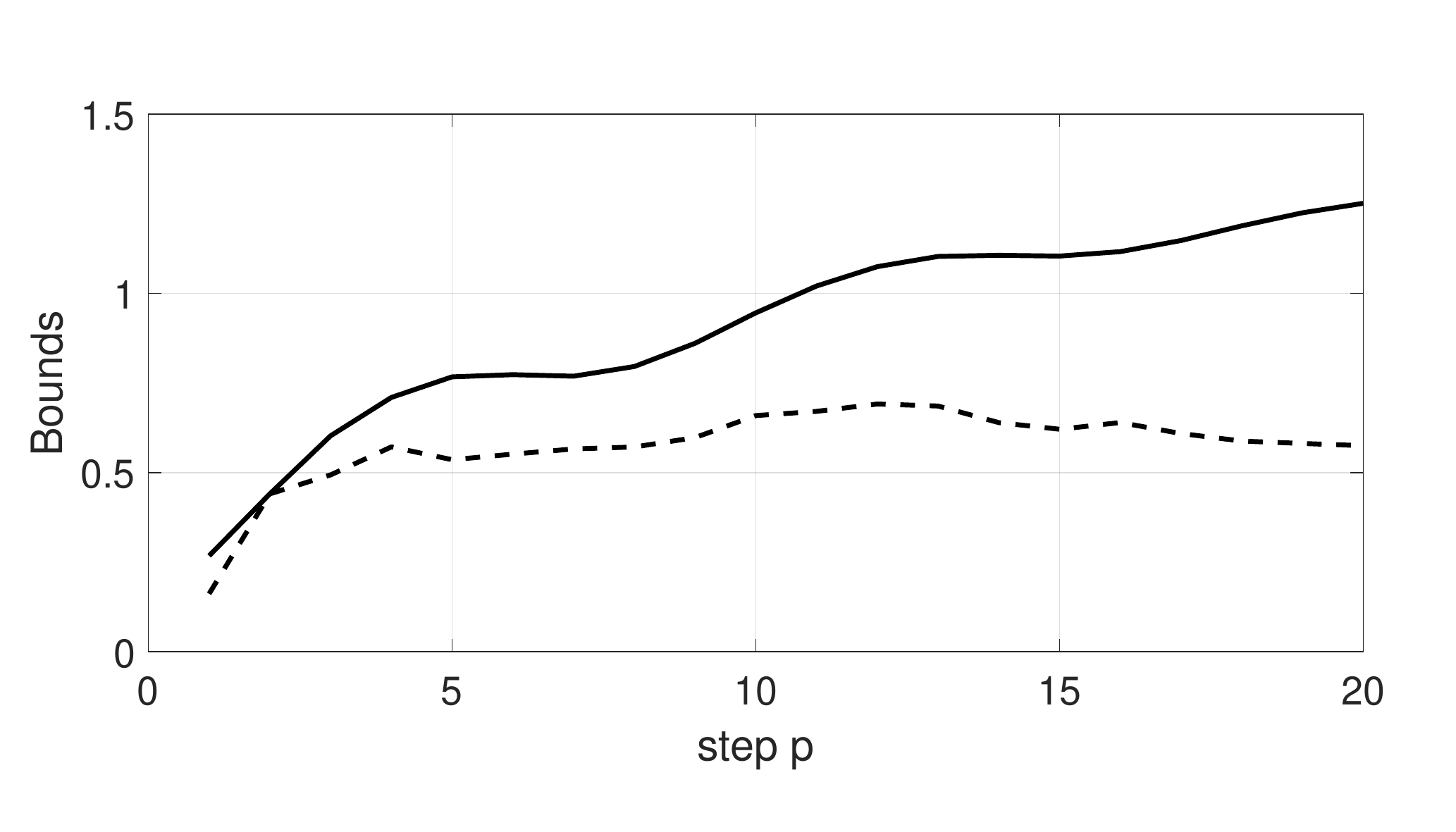}
\caption{Trend of bounds against step $p$. Solid line: bounds learned as in \eqref{eq:bound_ss_iteration} after optimizing $\bar{w}$, dashed line: multistep bounds associated to the iterated 1 step model $\hat{\tau}_p(\hat{\theta}^{(p),1})$}
\label{Figure_w_estimate_SM2}
\centering
\includegraphics[scale=0.39]{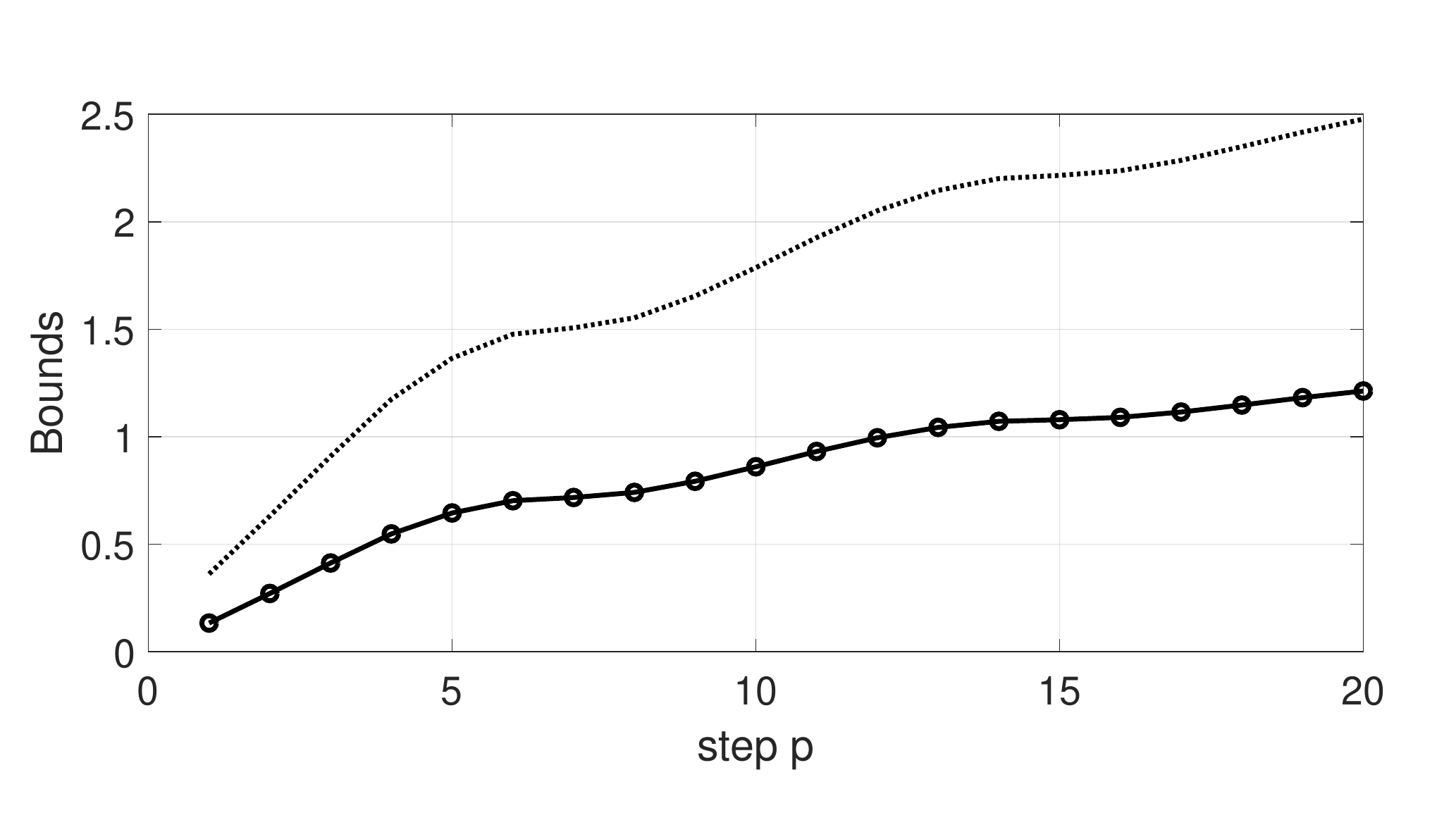}
\caption{Trend of bounds against step $p$. Dotted line: bounds with the previous approach we adopted (see \cite{TFFS18a}\cite{TFFS18b}) line with circles: current bounds for the additive disturbance $\sum_{i=0}^{p-1}| CA^iM_1 |\bar{w}$}
\label{Figure_w_estimate_SM}
\centering
\includegraphics[scale=0.39]{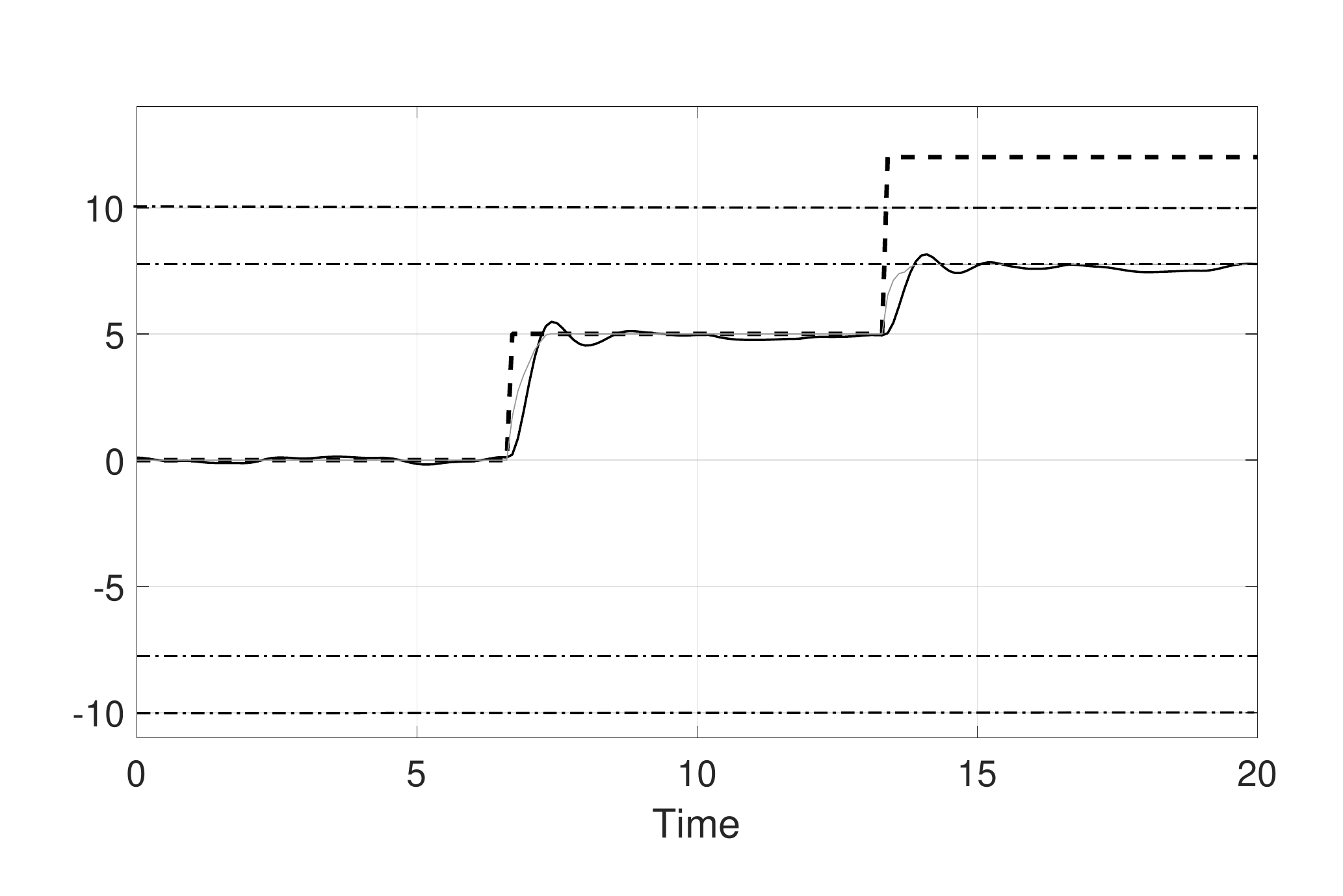}
\caption{Controlled system trajectories. Solid dark line:  $z(k)$, dashed line: reference value $z_{\rm\scriptscriptstyle goal}$, solid light line: $\bar{z}_0(k|k)$, dashed dot line: constraints and tightened constraints}
\label{Figure:closed_loop_z}
\centering
%\includegraphics[scale=0.40]{u_controlled}
%\caption{Controlled system trajectories. solid dark line $u(k)$, solid light line $\bar{u}(k|k)$ dashed lines: constraints and tightened constraints}
%\label{Figure:closed_loop_u}
\end{figure}

The multi-step guaranteed bounds $\hat{\tau}_p(\hat{\theta}^{(p),1})$ are compared with the bounds \eqref{eq:bound_ss_iteration}, pertaining to the state-space model with matrices \eqref{eq:matrices} in Figure \ref{Figure_w_estimate_SM2}. The multi-step bounds are smaller, as expected, however they are decoupled in time, and therefore not directly usable to guarantee recursive feasibility within the proposed robust MPC law.%which result to be the least conservative ones, and are indeed needed to solve optimization problem \eqref{eq:wcomputation_dknown}.

To clarify the advantage of our approach, in Figure \ref{Figure_w_estimate_SM} we compare the guaranteed bounds computed by iterating the obtained state-space model as described in this paper with those achieved by considering the uncertainty bound $\bar{w}=\hat{\tau}_1(\hat{\theta}^{(1)*})+\bar{d}$, i.e. the one-step-ahead guaranteed prediction error bound, iterating it over time with the same model matrices, and eventually adding $\bar{d}$. 
This alternative bound has been proposed in out previous works \cite{TFFS18a} and \cite{TFFS18b}.
%, and it is here briefly recalled. We consider the state space model $$\begin{cases}
%X_y(k+1)=AX_y(k)+B_1u(k)+M_1\tilde{w}(k)\\
%y(k)=CX_y(k)\\
%\end{cases}
%$$ 
%With standard manipulations we write $|\tilde{w}(k)|=|y(k+1)-\hat{z}^{(1)}(k+1)|\leq \hat{\tau}_1(\hat{\theta}^{(1)*})+\bar{d}$, see also \eqref{eq:predictor_k_p} and \eqref{eq:bound_ss_multistep} with $p=1$, thus $|y(k+p)-\hat{z}^{(1)}(k+p)| \leq C\sum_{i=0}^{p-1}{A^iM_1\left( \hat{\tau}_1(\hat{\theta}^{(1)*})+\bar{d}\right)}$ by integration of uncertainty bound, and eventually add $\bar{d}$ since $|z(k+p)-\hat{z}^{(1)}(k+p)|\leq |y(k+p)-\hat{z}^{(1)}(k+p) - d(k+p)|\leq |y(k+p)-\hat{z}^{(1)}(k+p)| + \bar{d}$.
It can be noted that the proposed approach achieves a guaranteed bound on the prediction error that is half the one obtained from the integration of  $\hat{\tau}_1(\hat{\theta}^{(1)*})+\bar{d}$, thus reducing conservativeness significantly.
%as state  for the uncertainty model proposed in this paper is computed as $ \sum_{i=0}^{p-1}| CA^iM_1 |\bar{w}+\|CA^{p}E\|_{\infty}\bar{d}$, see \eqref{eq:bound_ss_iteration}. It is compared with $\sum_{i=0}^{p-1}|C{A^iM_1|(\hat{\tau}_1(\hat{\theta}^{(1)*}) + \bar{d} )} + \bar{d} $, that represents an alternative uncertainty bound, see \cite{TFFS_CDC18} for the complete discussion of the two.

In the control design phase, the constraint sets $ \mathbb{U}=\mathbb{Z}= [-10,10]$ are considered, while the prediction and control horizon is $\bar{p}=10$.
The Luenberger observer and the auxiliary control law are chosen thanks to optimal control theory and the weighting matrices are tuned according to \eqref{eq:LMI_TOT}. The reference to be tracked is piece-wise constant and takes value $\{0,5,12\}$, thus including an unfeasible setpoint as well.
The trajectories of the closed-loop system are reported in Figure \ref{Figure:closed_loop_z}, where it is shown that $\bar{z}_0(k|k) \to z_{\rm\scriptscriptstyle goal}$ or to its nearest feasible point. As visible from the simulations, the infeasibile reference is handled successfully by the controller as well as the transients with respect to the open loop response of the system.
\section{Conclusions} \label{sec:conclusions}
The proposed unitary approach to learning-based MPC for linear systems allows one to design a  control law based on a dataset collected from the working plant. The obtained data-driven controller is able to effectively deal with constraints and track desired output references. The method relies on two phases: model learning and model-based control design, that are conceived to limit conservativeness while still robustly guaranteeing constraint satisfaction. To achieve this result, multi-step predictors and the related uncertainty bounds are derived and exploited to compute the state-space model employed in the MPC design.
Future directions are concerned with the extension to classes of nonlinear systems, the online (adaptive) computation of the prediction models and disturbance bounds, and the direct use of multi-step predictors also in the constraint tightening scheme.
\section*{Appendix} \label{sec:proofs} %(\hl{Warning: phi theta -> theta phi})
\small
\noindent \emph{Proof of Theorem \ref{convergence}}\\
%Problem \eqref{epsilon^o} can be rewritten as:
%\begin{equation}
%\begin{array}{ll}
%\bar{\epsilon}_p^*=&\min\limits_{\theta^{(p)} \in \Omega^{(p)}, \epsilon_p \in \mathbb{R}} \epsilon_p\\
%&\text{subject to}\\
%&|y_p-\theta^{(p)^T}\varphi_{y}^{(p)}|\leq\epsilon_p+\overline{d},\,\forall (\varphi_{y}^{(p)},y_p):\left[\begin{array}{c} \varphi_{y}^{(p)^T}y_p\end{array} \right]^T\in\mathcal{T}_p
%\end{array}
%\label{eq:reformulation1}
%\end{equation}
%At first, note that 
\noindent\emph{Proof of claim 1)}. The solution to \eqref{epsilon^o} must imply a regressor, denoted with $\varphi_{y,0}^{(p)}$, and a corresponding output value, $y_{p,0}$, satisfying:
\begin{equation}\label{epsilonbar0}
\begin{array}{rcl}
\bar{\epsilon}^*_p&=&|y_{p,0}-\tilde{\bar{\theta}}^{(p)^T} \varphi_{y,0}^{(p)}| - \bar{d}, \text{ so that, see \eqref{theta^o}}\\
\bar{\epsilon}^*_p&\geq&| y_{p,0}-\theta^{(p)^T}\varphi_{y,0}^{(p)}  | - \bar{d},\,\forall (\varphi_{y}^{(p)},y_p):\left[\begin{array}{c} \varphi_{y}^{(p)^T}y_p\end{array} \right]^T\in\mathcal{T}_p
\end{array}
\end{equation}
From the definition of $\mathcal{T}_p$ and $\tilde{\mathcal{T}}^{N_p}_p$ it holds that $\tilde{\mathcal{T}}^{N_p}_p \subset \mathcal{T}_p$; thus, from \eqref{E:bound_estim} neglecting the trivial case $\underline{\lambda}_p=0$, we have $
\underline{\lambda}_p = \min\limits_{\theta^{(p)} \in \Omega^{(p)}}\max\limits_{\left[\begin{array}{c} \tilde{\varphi}_{y}^{(p)^T} \tilde{y}_p\end{array} \right]^T\in\tilde{\mathcal{T}}_p^{N_p}} \lvert \tilde{y}_p - \theta^{(p)^T}\tilde{\varphi}_{y}^{(p)} \rvert - \bar{d} \leq \\
\min\limits_{\theta^{(p)} \in \Omega^{(p)}}\max\limits_{\left[\begin{array}{c} \varphi_{y}^{(p)^T} y_p\end{array} \right]^T\in\mathcal{T}_p} \lvert y_p -\theta^{(p)^T} \varphi_{y}^{(p)} \rvert - \bar{d}=\bar{\epsilon}^*_p$, thus we have then
$\underline{\lambda}_p  \leq \bar{\epsilon}_p^* $.\\
$\,$\\
\noindent\emph{Proof of claim 2)}. Starting from \eqref{E:bound_estim}, with standard properties of absolute values, we compute
\begin{equation}
\begin{split}
\underline{\lambda}_p(N_p)=&\max \biggl\{0,\min\limits_{\theta^{(p)} \in\Omega^{(p)}} \max\limits_{\left[\begin{array}{c} \tilde{\varphi}_{y}^{(p)^T} \tilde{y}_p\end{array} \right]^T\in\tilde{\mathcal{T}}_p^{N_p}}\lvert \tilde{y}_p- \theta^{(p)^T}\tilde{\varphi}_{y}^{(p)}\rvert - \bar{d}\biggr\} \geq \\
&\min\limits_{\theta^{(p)} \in\Omega^{(p)}}\left\{\lvert \bar{y}_p (N_p)-\theta^{(p)^T}\bar{\varphi}_{y}^{(p)}(N_p) \rvert - \bar{d} \right \} \text{, where}
\end{split}
\label{ineqlowerbound}
\end{equation}
\begin{equation}\nonumber
\begin{bmatrix}\bar{\varphi}_{y}^{(p)}(N_p)  \\ \bar{y}_p(N_p)\end{bmatrix}= \arg\min\limits_{\left[\begin{array}{c} \tilde{\varphi}_{y}^{(p)^T} \tilde{y}_p\end{array} \right]^T\in\tilde{\mathcal{T}}_p^{N_p}} \left\|\left[\begin{array}{c} \varphi_{y,0}^{(p)}\\ y_{p,0}\end{array} \right]-\left[\begin{array}{c} \tilde{\varphi}_{y}^{(p)}\\ \tilde{y}_p\end{array}\right]  \right\|_2
\label{eq:mindist}
\end{equation}
Adding and subtracting $y_{p,0}$ and $\theta^{(p)^T}\varphi_{y,0}^{(p)}$ from \eqref{ineqlowerbound}, and neglecting the trivial case $\underline{\lambda}_p=0$:
\begin{equation}
\begin{split}
\underline{\lambda}_p(N_p) \geq &\min\limits_{\theta^{(p)} \in\Omega^{(p)}}\{\lvert (\bar{y}_p(N_p)-y_{p,0}) + \\
&\theta^{(p)^T}( \varphi_{y,0}^{(p)}-\bar{\varphi}_{y}^{(p)}(N_p) )+ y_{p,0} - \theta^{(p)^T}\varphi_{y,0}^{(p)} \rvert - \bar{d}\} \\
\geq &\min\limits_{\theta^{(p)} \in\Omega^{(p)}}\{\lvert  y_{p,0} -\theta^{(p)^T}\varphi_{y,0}^{(p)}\rvert - \lvert(-\bar{y}_p(N_p)+y_{p,0}) - \\
&\theta^{(p)^T}(\varphi_{y,0}^{(p)}-\bar{\varphi}_{y}^{(p)}(N_p)) \rvert - \bar{d}\} \\
 =&\min\limits_{\theta^{(p)} \in\Omega^{(p)}}\{\lvert  y_{p,0} - \theta^{(p)^T}\varphi_{y,0}^{(p)} \rvert \}-\max\limits_{\theta^{(p)} \in\Omega^{(p)}} \{\lvert (-\bar{y}_p(N_p) + y_{p,0}) +\\
& \theta^{(p)^T}(-\varphi_{y,0}^{(p)}+\bar{\varphi}_{y}^{(p)}(N_p)) \rvert - \bar{d}\} \\
 =&\bar{\epsilon}_p^* + \bar{d}-
\max\limits_{\theta^{(p)} \in\Omega^{(p)}} \{\lvert (-\bar{y}_p(N_p) + y_{p,0}) +\\
& \theta^{(p)^T}(-\varphi_{y,0}^{(p)}+\bar{\varphi}_{y}^{(p)}(N_p)) \rvert - \bar{d}\}
\end{split}
\end{equation}
Then, simplifying $\bar{d}$, we compute
\begin{equation}
\begin{split}
\underline{\lambda}_p(N_p) \geq \bar{\epsilon}_p^* - \max\limits_{\theta^{(p)} \in\Omega^{(p)}} &\{\lvert (-\bar{y}_p(N_p) + y_{p,0}) \\
&+\theta^{(p)^T} (-\varphi_{y,0}^{(p)}+\bar{\varphi}_{y}^{(p)}(N_p)) \rvert\}
\end{split}
\end{equation}
Considering Assumption \ref{A:data_set}, we know that for $N_p \to \infty$
\begin{equation}
\forall \beta>0, \exists \bar{N}_p(\beta): \left\| \begin{array}{c}\bar{\varphi}_{y}^{(p)}(\bar{N_p})-\varphi_{y,0}^{(p)}\\ \bar{y}_{p}(\bar{N}_p)-y_{p,0}\end{array} \right\|\leq \beta, \text{in particular}\end{equation}
\begin{equation}
\|\bar{\varphi}_{y}^{(p)}(\bar{N_p})-\varphi_{y,0}^{(p)}\|\leq \beta \text{ and } \|\bar{y}_{p}(\bar{N}_p)-y_{p,0}\|\leq \beta
\label{eq:dist_points}
\end{equation}
and thus, using the inequality $\lvert a^Tb\rvert \leq \|a\|_2\|b\|_2$
\begin{equation}\nonumber
\begin{array}{ll}
\underline{\lambda}_p(\bar{N}_p(\beta)) &\geq\bar{\epsilon}_p^* - \max\limits_{\theta^{(p)} \in\Omega^{(p)}}\{\lvert (-\bar{y}_p(\bar{N}_p(\beta)) + y_{p,0}) \rvert \\
&+\|\theta^{(p)} \|_2  \|\varphi_{y,0}^{(p)}-\bar{\varphi}_{y}^{(p)}(\bar{N}_p(\beta)) \|_2 \} \\
 &\geq\bar{\epsilon}_p^* - \beta \left(1+ \max\limits_{\theta^{(p)} \in \Omega^{(p)}} \|\theta^{(p)}\|_2 \right)
\label{eq:lowerbound}
\end{array}
\end{equation}
Finally, choosing $\beta \leq \frac{\rho}{\left(1+ \max\limits_{\theta^{(p)} \in \Omega^{(p)}} \|\theta^{(p)}\|_2\right) } $ concludes the proof.

\noindent \emph{Proof of Lemma \ref{convergence2}}\\
The proof of claim 1 is very similar to that of claim 1 of Theorem \ref{convergence} and thus omitted for brevity. Regarding the second claim, we note that the definition of $\tau_p(\theta^{(p)})$ must imply a vector, denoted $\varphi_{y,0}^{(p)} \in \Phi^{(p)}$, that satisfies
\begin{equation}\nonumber
\tau_p(\theta^{(p)})-\hat{\bar{\epsilon}}_p=\max\limits_{\theta \in \Theta^{(p)}} |(\theta-\theta^{(p)})^T\varphi_{y,0}^{(p)}|
\end{equation}
Let us denote $\bar{\varphi}_{y}^{(p)}(N_p)= \arg\min\limits_{\begin{bmatrix}\tilde{\varphi}_{y}^{(p)^T}  \tilde{y}_p \end{bmatrix}^T \in \tilde{\mathcal{T}}_p^{N_p}} \| \tilde{\varphi}_{y}^{(p)}-\varphi_{y,0}^{(p)}\|_2$ and
\begin{equation}\nonumber
\begin{array}{cl}
\underline{\tau}_p(\theta^{(p)})-\hat{\bar{\epsilon}}_p &= \max\limits_{\begin{bmatrix}\tilde{\varphi}_{y}^{(p)^T}  \tilde{y}_p \end{bmatrix}^T \in \tilde{\mathcal{T}}_p^{N_p}} \max\limits_{\theta \in \Theta^{(p)}} |( \theta - \theta^{(p)})^T\tilde{\varphi}_{y}^{(p)}| \\
& \geq \max\limits_{\theta \in \Theta^{(p)}} |(\theta - \theta^{(p)})^T\bar{\varphi}_{y}^{(p)}(N_p)|
\end{array}
\end{equation}
By adding and subtracting $(\theta-\theta^{(p)^T})\varphi_{y,0}^{(p)}$, one can write
\begin{equation}\nonumber
\begin{array}{cl}
\underline{\tau}_p(\theta^{(p)})-\hat{\bar{\epsilon}}_p\geq &\\
\max\limits_{\theta \in \Theta^{(p)}} |(\theta-\theta^{(p)})^T \varphi_{y,0}^{(p)}-(\theta-\theta^{(p)})^T(\varphi_{y,0}^{(p)}-\bar{\varphi}_{y}^{(p)}(N_p)) | \geq \\
\max\limits_{\theta \in \Theta^{(p)}} |(\theta-\theta^{(p)})^T \varphi_{y,0}^{(p)}| - \max\limits_{\theta \in \Theta^{(p)}}|(\theta-\theta^{(p)})^T(\varphi_{y,0}^{(p)}-\bar{\varphi}_{y}^{(p)}(N_p))|
\end{array}
\end{equation}
Recall from Assumption \ref{A:data_set} that for $N_p \to \infty$
\begin{equation} \nonumber
\forall \beta>0, \exists \bar{N}_p(\beta) : \|\bar{\varphi}_{y}^{(p)}(\bar{N}_p(\beta))-\varphi_{y,0}^{(p)}\| \leq \beta
\end{equation}
and, using the property $|a^Tb| \leq \|a\|_2\|b\|_2$, also replacing $\max\limits_{\theta \in \Theta^{(p)^T}} |(\theta-\theta^{(p)}) \varphi_{y,0}^{(p)}|=\tau_p(\theta^{(p)})-\hat{\bar{\epsilon}}_p$, we eventually obtain:
\begin{equation}\nonumber
\begin{array}{cl}
\underline{\tau}_p(\theta^{(p)}) & \geq \tau_p(\theta^{(p)}) - \max\limits_{\theta \in \Theta^{(p)}}\|(\varphi_{y,0}^{(p)}-\bar{\varphi}_{y}^{(p)}(N_p))\|_2 \|(\theta-\theta^{(p)})\|_2\\
& \geq \tau_p(\theta^{(p)})  - \beta \max\limits_{\theta \in \Theta^{(p)}} \|\theta - \theta^{(p)}\|_2
\end{array}
\end{equation}
Now, taking any $\rho \in (0,\tau_p(\theta^{(p)}))$ and choosing $\beta \leq \frac{\rho}{\max\limits_{\theta \in \Theta^{(p)}} \| \theta  - \theta^{(p)}\|_2}$ we have $\underline{\tau}_p(\theta^{(p)}) \geq \tau_p(\theta^{(p)})-\rho$, which concludes the proof.

\noindent \emph{Proof of Theorem~\ref{thm:res1}}\\
The proof of Theorem \ref{thm:res1} is here divided into the following steps:
\begin{itemize}
	\item Proof of recursive feasibility of the optimization problem \eqref{eq:optprb}.
	\item Proof that constraints \eqref{eq:constraints_YU} are satisfied.
	\item Proof of convergence.
\end{itemize}
\emph{Recursive feasibility.}\\
The proof is conducted by induction. Assume that, at instant $k$, a solution to the optimization problem \eqref{eq:optprb} exists. All constraints \eqref{eq:constraint_optE}-\eqref{eq:constraint_optF} are therefore verified by $z_{\rm\scriptscriptstyle ref}(k|k)$, the trajectories $\bar{X}(k+p|k)$, and $\bar{U}(k|k)=(\bar{u}(k|k),\dots,\bar{u}(k+\overline{p}|k))$. The input $u(k)$, actually applied to the system at time step $k$, is defined according to \eqref{eq:controllaw}.\\
At step $k+1$, $X(k+1)=AX(k)+B_1u(k)+M_1w(k)$ and $\hat{X}(k+1)=A\hat{X}(k)+B_1u(k)+M_1\hat{w}(k)+L(y(k)-C\hat{X}(k))$. 
We can show that a feasible, although possibly suboptimal, solution to \eqref{eq:optprb} at step $k+1$ can be defined as $
z_{\rm\scriptscriptstyle ref}(k+1|k)=z_{\rm\scriptscriptstyle ref}(k|k)$, $\bar{X}(k+1|k)$,$$\bar{U}(k+1|k)=\left( \bar{u}(k+1|k), \dots, \bar{u}(k+\overline{p}|k), u_{\rm\scriptscriptstyle ref}(k|k)+K(\bar{X}(k+\overline{p}+1|k)-X_{\rm\scriptscriptstyle ref}(k|k)) \right).$$
First of all, in view of the invariance of $\bar{\mathbb{E}}$,
$ \hat{X}(k+1)-\bar{X}(k+1|k)=(A+B_1K)(\hat{X}(k)-\bar{X}(k|k))+LC\hat{e}(k)+Ld(k)\in (A+B_1K)\bar{\mathbb{E}}\oplus LC\hat{\mathbb{E}}\oplus \mathbb{D}\subseteq \bar{\mathbb{E}}$, where $\mathbb{D}=[-\bar{d},\bar{d}]$.
Also, $\bar{u}(k+p|k) \in \bar{\mathbb{U}}$ in view of \eqref{eq:constraint_opt}, for all $p=1, \dots, \overline{p}$; also, $u_{\rm\scriptscriptstyle ref}(k|k)+K(\bar{X}(k+\overline{p}+1|k)-X_{\rm\scriptscriptstyle ref}(k|k)) \in \bar{\mathbb{U}}$ in view of \eqref{eq:constraint_optF} and of the definition of $\mathbb{O}_{\epsilon}$. Last $\hat{w}(k|k) \in \mathbb{W}$ at time k ensures $\hat{w}(k+1|k) \in \mathbb{W}$ since $z_{\rm \scriptscriptstyle ref}(k+1|k)=z_{\rm \scriptscriptstyle ref}(k|k)$ and \eqref{eq:what_zref}. \\
Thirdly, $C\bar{X}(k+p|k) \in \bar{\mathbb{Z}}$ for all $p=1,\dots,\overline{p}$ in view of \eqref{eq:constraint_opt} and $C\bar{X}(k+\overline{p}+1|k) \in \bar{\mathbb{Z}}$ in view of \eqref{eq:constraint_optF} and of the definition of $\mathbb{O}_{\epsilon}$. Finally, it holds that $$\begin{bmatrix}\bar{X}(k+\overline{p}+2|k)\\z_{\rm\scriptscriptstyle ref}(k+1|k)\end{bmatrix}
=\mathcal{F}\begin{bmatrix}\bar{X}(k+\overline{p}+1|k)\\z_{\rm\scriptscriptstyle ref}(k|k)\end{bmatrix}\in \mathbb{O}_{\epsilon}$$
in view of \eqref{eq:constraint_optF} and of the positive invariance of $\mathbb{O}_{\epsilon}$.
Since feasibility holds by assumption at time $k=0$ then, by induction, it is guaranteed also for all $k>0$.
\vspace{2 mm}\\
\noindent
\emph{Constraint satisfaction.}\\
In view of the feasibility of the problem \eqref{eq:optprb} at any time instant $k\geq0$, it results that constraints \eqref{eq:constraint_optE}-\eqref{eq:constraint_opt} are verified. Therefore, for all $k\geq 0$, from \eqref{eq:controllaw} $u(k)=\bar{u}(k|k)+K(\hat{X}(k)-\bar{X}(k|k))\in \bar{\mathbb{U}} \oplus K\bar{\mathbb{E}}\subseteq\mathbb{U}$, proving \eqref{eq:constraints_U}. Also, $z(k)=C X(k)=C\bar{X}(k|k)+C(\hat{X}(k)-\bar{X}(k|k))+C({X}(k)-\hat{X}(k|k))\in \bar{\mathbb{Z}}\oplus  C_0\bar{\mathbb{E}}\oplus  C_0\hat{\mathbb{E}}\subseteq \mathbb{Z}$, proving \eqref{eq:constraints_Y}.

\noindent
\emph{Convergence.}\\
We compute, from \eqref{eq:cost_fcn}, that
\begin{equation}
\begin{array}{ll}
J(k|k)&=\sum_{p=0}^{\overline{p}}\|\begin{bmatrix}C_p&D_p\end{bmatrix}\begin{bmatrix}\bar{X}(k|k)\\\bar{U}(k|k)\end{bmatrix}  - z_{\rm\scriptscriptstyle ref}^{p}(k)\|^2_{Q_p}\\
&+\|\bar{u}(k+p|k) - \textit{u}_{\rm\scriptscriptstyle ref}(k|k)\|^2_{R_p}\\
&+\|\bar{X}(k+\overline{p}+1|k) - X_{\rm\scriptscriptstyle ref}(k|k)\|^2_P + \sigma\| \textit{z}_{\rm\scriptscriptstyle ref}(k|k) - \textit{z}_{\rm\scriptscriptstyle goal} \|^2
\end{array}
\label{eq:Jkk}
\end{equation}
At step $k+1$, the optimal cost function is $J(k+1|k+1)\leq J(k+1|k)$, where $J(k+1|k)$ is the cost obtained if the feasible solution $\bar{X}(k+1|k), \bar{U}(k+1|k),\textit{z}_{\rm\scriptscriptstyle ref}(k|k)$ is applied. We compute that
\begin{equation}
\begin{array}{ll}
J(k+1|k)&=\sum_{p=0}^{\overline{p}}\|\begin{bmatrix}C_p&D_p\end{bmatrix}\begin{bmatrix}\bar{X}(k+1|k)\\\bar{U}(k+1|k)\end{bmatrix}  - z_{\rm\scriptscriptstyle ref}^p(k|k)\|^2_{Q_p}\\
&+\|\bar{u}(k+p+1|k) - \textit{u}_{\rm\scriptscriptstyle ref}(k|k)\|^2_{R_p}\\
&+\|\bar{X}(k+\overline{p}+2|k) - X_{\rm\scriptscriptstyle ref}(k|k)\|^2_P + \sigma\| \textit{z}_{\rm\scriptscriptstyle ref}(k|k) - \textit{z}_{\rm\scriptscriptstyle goal} \|^2
\end{array}
\label{eq:Jk1k}
\end{equation}
where $\bar{u}(k+\overline{p}+1|k)= u_{\rm\scriptscriptstyle ref}(k|k)+K\left( \bar{X}(k+\overline{p}+1|k)-X_{\rm\scriptscriptstyle ref}(k|k) \right)$.
From \eqref{eq:Jkk} and \eqref{eq:Jk1k},
$J(k+1|k+1)-J(k|k)\leq J(k+1|k)-J(k|k)\leq-(\|\bar{z}_0(k|k) - z_{\rm\scriptscriptstyle ref}(k|k)\|^2_{Q_0}+\|\bar{u}(k|k) - u_{\rm\scriptscriptstyle ref}(k|k)\|^2_{R_0})
\\+\sum_{p=0}^{\overline{p}-1}(\|\begin{bmatrix}C_p&D_p\end{bmatrix}\begin{bmatrix}\bar{X}(k+1|k)\\\bar{U}(k+1|k)\end{bmatrix} - z_{\rm\scriptscriptstyle ref}^{p}(k|k)\|^2_{Q_p}\break -
\|\begin{bmatrix}C_{p+1}&D_{p+1}\end{bmatrix}\begin{bmatrix}\bar{X}(k|k)\\\bar{U}(k|k)\end{bmatrix} - z_{\rm\scriptscriptstyle ref}^{p+1}(k|k)\|^2_{Q_{p+1}}+\|\bar{u}(k+1+p|k) - u_{\rm\scriptscriptstyle ref}(k|k)\|^2_{R_p}-\|\bar{u}(k+p+1|k) - u_{\rm\scriptscriptstyle ref}(k|k)\|^2_{R_{p+1}})+\|\begin{bmatrix}C_{\overline{p}}&D_{\overline{p}}\end{bmatrix}\begin{bmatrix}\bar{X}(k+1|k)\\
\bar{U}(k+1|k)\end{bmatrix} - {z}_{\rm\scriptscriptstyle ref}^{\bar{p}}(k)\|^2_{Q_{\overline{p}}}
+\|\bar{u}(k+\bar{p}+1|k)-u_{\rm\scriptscriptstyle ref}(k|k)\|^2_{R_{\overline{p}}}
-\|\bar{X}(k+1+{\overline{p}}|k) - X_{\rm\scriptscriptstyle ref}(k|k)\|^2_{P}+\|(A+B_1K)(\hat{X}(k+1+{\overline{p}}|k) - X_{\rm\scriptscriptstyle ref}(k|k) )\|^2_{P})$, where we used $\bar{X}(k+\bar{p}+2|k)-X_{\rm\scriptscriptstyle ref}(k|k)=(A+B_1K)(\bar{X}(k+1+{\overline{p}}|k) - X_{\rm\scriptscriptstyle ref}(k|k))$.\\
%This is justified by subtraction between $\bar{X}(k+\bar{p}+2|k)=A\bar{X}(k+\bar{p}+1|k) + B_1\left( u_{ref}+K(\bar{X}(k+\bar{p}+1|k) - X_{ref}) \right)$ and $X_{ref}=AX_{ref}+B_1u_{ref}$.
%
We can write that $\bar{U}(k+1|k)=H_1\bar{U}(k|k)+H_2 ( u_{\rm\scriptscriptstyle ref}(k|k) + K(\bar{X}(k+\overline{p}+1|k) - X_{\rm\scriptscriptstyle ref}(k|k) ) )$ and that $\mathbf{1}_{\bar{p}+1}u_{\rm\scriptscriptstyle ref}(k|k)=H_1\mathbf{1}_{\bar{p}+1}u_{\rm\scriptscriptstyle ref}(k|k) + H_2 u_{\rm\scriptscriptstyle ref}(k|k)$, being $H_2=\begin{bmatrix}0_{1,\overline{p}}&1\end{bmatrix}^T$. Also, $\bar{X}(k+1+{\overline{p}}|k) - X_{\rm\scriptscriptstyle ref}(k|k) = [A^{\bar{p}+1}\quad  \Gamma]\begin{bmatrix}\bar{X}(k|k) -X_{\rm\scriptscriptstyle ref}(k|k)\\\bar{U}(k|k)-\mathbf{1}_{\bar{p}+1}u_{\rm\scriptscriptstyle ref}(k|k)\end{bmatrix}$.
In view of this we can write
\begin{equation}
\begin{split}
&\begin{bmatrix}\bar{X}(k+1|k)-X_{\rm\scriptscriptstyle ref}(k|k)\\\bar{U}(k+1|k)-\mathbf{1}_{\bar{p}+1}u_{\rm\scriptscriptstyle ref}(k|k)\end{bmatrix}=\\
&\begin{bmatrix}A&B\\H_2 KA^{\overline{p}+1}&H_1+H_2 K\Gamma
\end{bmatrix}\begin{bmatrix}\bar{X}(k|k) -X_{\rm\scriptscriptstyle ref}(k|k)\\\bar{U}(k|k)-\mathbf{1}_{\bar{p}+1}u_{\rm\scriptscriptstyle ref}(k|k)\end{bmatrix}
\end{split}
\end{equation}
For notational simplicity, we define
$$\bar{\xi}(k|k)=\begin{bmatrix}\bar{X}(k|k) - X_{\rm\scriptscriptstyle ref}(k|k)\\\bar{U}(k|k) - \mathbf{1}_{\bar{p}+1}u_{\rm\scriptscriptstyle ref}(k|k)\end{bmatrix}$$
Defining $p\geq 1$, $\Delta R_p=R_p-R_{p-1}$ and recalling that $D_pH_2=0$, for all $p=1,\dots,\overline{p}$, we can write that
$J(k+1|k+1)-J(k|k) \leq \\
-(\|\bar{z}_0(k|k)-z_{\rm\scriptscriptstyle ref}(k|k)\|^2_{Q_0}+\|\bar{u}(k|k)-u_{\rm\scriptscriptstyle ref}(k|k)\|^2_{R_0/2})\\
+\sum_{p=0}^{\overline{p}-1} \biggl( \| \begin{bmatrix}C_pA&C_pB+D_pH_1\end{bmatrix}\bar{\xi}(k|k)\|_{Q_p}^2 -\\
\|\begin{bmatrix}C_{p+1}&D_{p+1}\end{bmatrix}\bar{\xi}(k|k)\|^2_{Q_{p+1}} \biggr)- \|\begin{bmatrix}0&I_{\overline{p}+1}\end{bmatrix}\bar{\xi}(k|k)\|^2_{\mathcal{R}}\\
+\|\begin{bmatrix}C_{\overline{p}}A&C_{\overline{p}}B+D_{\overline{p}}H_1\end{bmatrix}\bar{\xi}(k|k)\|_{Q_{\overline{p}}}^2\\
+ \|\bar{X}(k+1+{\overline{p}}|k) - X_{\rm\scriptscriptstyle ref}(k|k)\|^2_{K^TR_{\bar{p}}K+(A+B_1K)^TP(A+B_1K)-P}$.
From \eqref{eq:LMI1_p} we obtain that
$J(k+1|k+1)-J(k|k) \leq -(\|\bar{z}_0(k|k)-z_{\rm\scriptscriptstyle ref}(k|k)\|^2_{Q_0}+\|\bar{u}(k|k)-u_{\rm\scriptscriptstyle ref}(k|k)\|^2_{R_0/2})\\
+\sum_{p=0}^{\overline{p}-1} \biggl( \| \begin{bmatrix}C_pA&C_pB+D_pH_1\end{bmatrix}\bar{\xi}(k|k)\|_{Q_p}^2  -\\
\|\begin{bmatrix}C_{p+1}&D_{p+1}\end{bmatrix}\bar{\xi}(k|k)\|^2_{Q_{p+1}} \biggr)
-\|\begin{bmatrix}0_{\bar{p}+1,2o-1}&I_{\overline{p}+1}\end{bmatrix}\bar{\xi}(k|k)\|^2_{\mathcal{R}}
+\| \begin{bmatrix}C_{\overline{p}}A&C_{\overline{p}}B+D_{\overline{p}}H_1\end{bmatrix}\bar{\xi}(k|k)\|^2_{Q_{\overline{p}}}-\|\bar{X}(k+1+{\overline{p}}|k) - X_{\rm\scriptscriptstyle ref}(k|k)\|^2_{T_N}$.\\
We can write the latter in compact form as
\begin{equation}
\begin{split}
&J(k+1|k+1)-J(k|k) \leq  -(\|\bar{z}_0(k|k)-z_{\rm\scriptscriptstyle ref}(k|k)\|^2_{Q_0}\\
&+\|\bar{u}(k|k)-u_{\rm\scriptscriptstyle ref}(k|k)\|^2_{R_0/2})+
\|\bar{\xi}(k|k)\|^2_{\Omega}
\end{split}
\label{eq:iss_form}
\end{equation}
where $\Omega= \Psi^T\mathcal{Q}\Psi-\bar{\Psi}^T\bar{\mathcal{Q}}\bar{\Psi}$.
In view of \eqref{eq:LMI2_p}
\begin{equation}
\begin{split}
J(k+1|k+1)&-J(k|k) \leq  -(\|\bar{z}_0(k|k)-z_{\rm\scriptscriptstyle ref}(k|k)\|^2_{Q_0}\\
&+\|\bar{u}(k|k)-u_{\rm\scriptscriptstyle ref}(k|k)\|^2_{R_0/2})
\end{split}
\label{eq:iss_form2}
\end{equation}
In view of \eqref{eq:iss_form2} then, asymptotically, $\|\bar{z}_0(k|k)-z_{\rm\scriptscriptstyle ref}(k|k)\|^2_{Q_0}+\|\bar{u}(k|k)-u_{\rm\scriptscriptstyle ref}(k|k)\|^2_{R_0/2}\rightarrow 0$ as $k \to +\infty$. \medskip\\
In the final part of the proof we show that, similarly to~\cite{limon2008automatica} the only asymptotic solution compatible with $\|\bar{z}_0(k|k)-z_{\rm\scriptscriptstyle ref}(k|k)\|^2_{Q_0}=0$ and $\|\bar{u}(k|k)-u_{\rm\scriptscriptstyle ref}(k|k)\|^2_{R_0/2}= 0$ is the one corresponding to $z_{\rm\scriptscriptstyle ref}(k|k)=z_{\rm\scriptscriptstyle goal}^{{\rm\scriptscriptstyle FEASIBLE}}$.
Similarly to \cite{limon2008automatica} we proceed by contradiction.

Preliminarly we highlight a property of the MOAS, specifically that the definition \eqref{eq:O_def} implies, for all $\bar{X}$,
$\lim\limits_{k\rightarrow +\infty}\mathcal{C}\mathcal{F}^k(\bar{X},z_{\rm\scriptscriptstyle ref})=\lim\limits_{k\rightarrow +\infty}\mathcal{C}\mathcal{F}^k(X_{\rm\scriptscriptstyle ref},z_{\rm\scriptscriptstyle ref})=\mathcal{C}(X_{\rm\scriptscriptstyle ref},z_{\rm\scriptscriptstyle ref})=(z_{\rm\scriptscriptstyle ref},u_{\rm\scriptscriptstyle ref},\hat{w}_{ref})$
where
$$\begin{bmatrix}X_{\rm\scriptscriptstyle ref}& u_{\rm\scriptscriptstyle ref} & \hat{w}_{ref} \\ \end{bmatrix}^T=\begin{bmatrix} N & \hat{\mu}^{-1} &  \eta_{zw}\end{bmatrix}^Tz_{\rm\scriptscriptstyle ref}$$
$$\text{and so }\lim_{k\rightarrow +\infty}\mathcal{C}\mathcal{F}^k(\bar{X},z_{\rm\scriptscriptstyle ref})=\begin{bmatrix} CN & \hat{\mu}^{-1} & \eta_{zw}\end{bmatrix}^Tz_{\rm\scriptscriptstyle ref}$$ Therefore, we assume by contradiction that $z_{\rm\scriptscriptstyle ref}^{\infty}\neq z_{\rm\scriptscriptstyle goal}^{{\rm\scriptscriptstyle FEASIBLE}}$,
where
$$\begin{bmatrix} CN & \hat{\mu}^{-1} & \eta_{zw}\end{bmatrix}^Tz_{\rm\scriptscriptstyle ref}^{\infty}\in\bar{\mathbb{X}}_{\mathbb{ZUW}}(\varepsilon)$$
is the steady-state solution where the output $\bar{z}_0(k|k)$ converges. In steady state, $\bar{z}_0(k|k) = z_{\rm\scriptscriptstyle ref}^{\infty}$, corresponding to the state and input values
$$\begin{bmatrix}\bar{X}(k)\\\bar{u}(k)\end{bmatrix}=\begin{bmatrix}{X}_{\rm\scriptscriptstyle ref}^{\infty}\\{u}_{\rm\scriptscriptstyle ref}^{\infty}\end{bmatrix}=\begin{bmatrix}N\\ \hat{\mu}^{-1}\end{bmatrix}{z}^{\infty}_{\rm\scriptscriptstyle ref}$$
Note that such steady-state condition is feasible for \eqref{eq:optprb}, since all constraints are verified, i.e., \eqref{eq:constraint_optE}-\eqref{eq:constraint_opt}. The corresponding value of the cost function \eqref{eq:Jkk} is $J_1=\sigma \|z_{\rm\scriptscriptstyle ref}^{\infty}-z_{\rm\scriptscriptstyle goal} \|^2$.\\
Consider now an alternative solution (starting, at time $k$, from the previously-defined steady-state) to problem~\eqref{eq:optprb}, i.e., the triple
$(\bar{X}(k),\bar{U}(k),z_{\rm\scriptscriptstyle ref}(k))$ where
the initial condition $\bar{X}(k)$, compatible with constraint \eqref{eq:constraint_optE}, is $\bar{X}(k)={X}_{\rm\scriptscriptstyle ref}^{\infty}$ and
the reference output is
\begin{equation}
{z}_{\rm\scriptscriptstyle ref}=z_{\rm\scriptscriptstyle ref}^{\infty}\lambda + (1-\lambda)z_{\rm\scriptscriptstyle goal}^{{\rm\scriptscriptstyle FEASIBLE}}
\label{eq:ytilde_ref}
\end{equation}
that, in passing, corresponds to the following reference values for $\bar{X}$ and $\bar{u},\begin{bmatrix}{X}_{\rm\scriptscriptstyle ref}\\{u}_{\rm\scriptscriptstyle ref}\end{bmatrix}=\begin{bmatrix}N\\ \hat{\mu}^{-1}\end{bmatrix}{z}_{\rm\scriptscriptstyle ref}$.
Finally, the alternative input sequence is given by
$\bar{u}(k+p)={u}_{\rm\scriptscriptstyle ref} + K(\bar{X}(k+p) - {X}_{\rm\scriptscriptstyle ref})$. Note that, importantly, the latter alternative solution to~\eqref{eq:optprb} is feasible (i.e., also verifies \eqref{eq:constraint_opt}) if $(1-\lambda)$ is sufficiently small (with $\lambda\neq 1$). The corresponding cost function $J_2$ reads:
\begin{equation}
\begin{array}{ll}
J_2&=\sum_{p=0}^{\overline{p}}\|\begin{bmatrix}C_p&D_p\end{bmatrix}\begin{bmatrix}\bar{X}(k) - {X}_{\rm\scriptscriptstyle ref}\\ \bar{U}(k)-\mathbf{1}_{\bar{p}+1,1}{u}_{\rm\scriptscriptstyle ref}\end{bmatrix}\|^2_{Q_p}
+\|\bar{u}(k+p) - u_{\rm\scriptscriptstyle ref}\|^2_{R_p}\\
&+\|\bar{X}(k+\overline{p}+1) - {X}_{\rm\scriptscriptstyle ref}\|^2_P + \sigma\| z_{\rm\scriptscriptstyle ref} - \textit{z}_{\rm\scriptscriptstyle goal} \|^2\\
&=\sum_{p=0}^{\overline{p}}\|\begin{bmatrix}C_p&D_p\end{bmatrix}G_{xu}(z_{\rm\scriptscriptstyle ref}^{\infty}-{z}_{\rm\scriptscriptstyle ref})\|^2_{Q_p}\\
&+\|\left[ K (A+B_1K)^p \quad 0_{1,\bar{p}+1}\right]G_{xu}(z_{\rm\scriptscriptstyle ref}^{\infty}-{z}_{\rm\scriptscriptstyle ref})\|^2_{R_p}\\
&+\| \left[  (A+B_1K)^{\bar{p}+1} \quad 0_{1,\bar{p}+1} \right] G_{xu}(z_{\rm\scriptscriptstyle ref}^{\infty}-{z}_{\rm\scriptscriptstyle ref})\|^2_P + \sigma\| z_{\rm\scriptscriptstyle ref} - z_{\rm\scriptscriptstyle goal} \|^2\\
&=\|z_{\rm\scriptscriptstyle ref}^{\infty}-{z}_{\rm\scriptscriptstyle ref}\|_{\tilde{P}}^2 + \sigma \|z_{\rm\scriptscriptstyle ref} - z_{\rm\scriptscriptstyle goal}\|^2
\end{array}
\end{equation}
where
$$G_{xu}=\begin{bmatrix}I_{2o-1}&0_{2o-1,1}\\0_{\bar{p}+1,2o-1}&\mathbf{1}_{\bar{p}+1}\end{bmatrix}\begin{bmatrix}N\\ \hat{\mu}^{-1}\end{bmatrix}$$
and where $\tilde{P}=G_{xu}^T\Lambda^{T} $diag$(Q_0,\dots,Q_{\bar{p}},R_0,\dots ,R_p,P) \Lambda G_{xu}$ with
\begin{equation}
\Lambda=\begin{bmatrix}
C_0 & D_0\\
\vdots & \vdots\\
C_{\bar{p}} & D_{\bar{p}}\\
K(A+B_1K)^{0} & 0_{1,\bar{p}+1}\\
\vdots \\
K(A+B_1K)^{\bar{p}}&0_{1,\bar{p}+1}\\
(A+B_1K)^{\bar{p}+1}&0_{2o-1,\bar{p}+1}
\end{bmatrix}
\end{equation}
From \eqref{eq:ytilde_ref}, $z_{\rm\scriptscriptstyle ref}^{\infty}-{z}_{\rm\scriptscriptstyle ref}=(1-\lambda)(z_{\rm\scriptscriptstyle ref}^{\infty}-{z}_{\rm\scriptscriptstyle goal}^{{\rm\scriptscriptstyle FEASIBLE}})$ and ${z}_{\rm\scriptscriptstyle ref}-z_{\rm\scriptscriptstyle goal}=\lambda (z_{\rm\scriptscriptstyle ref}^{\infty} -z_{\rm\scriptscriptstyle goal}^{{\rm\scriptscriptstyle FEASIBLE}})+z_{\rm\scriptscriptstyle goal}^{{\rm\scriptscriptstyle FEASIBLE}}-z_{\rm\scriptscriptstyle goal}$. If $\sigma$ is sufficiently large, i.e. if $\sigma > \lambda_{max}(\tilde{P})$ we compute that
$$\begin{array}{lcl}
J_1/\sigma & = & \|z_{\rm\scriptscriptstyle ref}^{\infty}-{z}_{\rm\scriptscriptstyle ref}^{{\rm\scriptscriptstyle FEASIBLE}}\|^2+\|{z}_{\rm\scriptscriptstyle ref}^{{\rm\scriptscriptstyle FEASIBLE}}-{z}_{\rm\scriptscriptstyle ref}\|^2\\
&&+2({z}_{\rm\scriptscriptstyle goal}^{{\rm\scriptscriptstyle FEASIBLE}}-{z}_{\rm\scriptscriptstyle goal})^T(z_{\rm\scriptscriptstyle ref}^{\infty}-{z}_{\rm\scriptscriptstyle goal}^{{\rm\scriptscriptstyle FEASIBLE}})\\
J_2/\sigma &\leq&  \|z_{\rm\scriptscriptstyle ref}^{\infty}-{z}_{\rm\scriptscriptstyle ref}\|^2 + \|{z}_{\rm\scriptscriptstyle ref}-z_{\rm\scriptscriptstyle goal}\|^2\\
&\leq&\left((1-\lambda)^2+\lambda^2\right)\|z_{\rm\scriptscriptstyle ref}^{\infty}-{z}_{\rm\scriptscriptstyle goal}^{{\rm\scriptscriptstyle FEASIBLE}}\|^2\\
&&+\|{z}_{\rm\scriptscriptstyle ref}^{{\rm\scriptscriptstyle FEASIBLE}}-{z}_{\rm\scriptscriptstyle goal}\|^2\\&&+2\lambda({z}_{\rm\scriptscriptstyle goal}^{{\rm\scriptscriptstyle FEASIBLE}}-{z}_{\rm\scriptscriptstyle goal})^T(z_{\rm\scriptscriptstyle ref}^{\infty}-{z}_{\rm\scriptscriptstyle goal}^{{\rm\scriptscriptstyle FEASIBLE}})
\end{array}$$
By subtraction we get
\begin{equation}
\begin{split}
J_1 - J_2 &\geq (1-(1-\lambda)^2 - \lambda^2)\sigma \|z_{\rm\scriptscriptstyle ref}^{\infty}-{z}_{\rm\scriptscriptstyle goal}^{{\rm\scriptscriptstyle FEASIBLE}}\|^2\\
&+2(1-\lambda)\sigma({z}_{\rm\scriptscriptstyle goal}^{{\rm\scriptscriptstyle FEASIBLE}}-{z}_{\rm\scriptscriptstyle goal})^T(z_{\rm\scriptscriptstyle ref}^{\infty}-{z}_{\rm\scriptscriptstyle goal}^{{\rm\scriptscriptstyle FEASIBLE}})
\end{split}
\end{equation}
Since  $2 ({z}_{\rm\scriptscriptstyle goal}^{{\rm\scriptscriptstyle FEASIBLE}}-{z}_{\rm\scriptscriptstyle goal})^T(z_{\rm\scriptscriptstyle ref}^{\infty}-{z}_{\rm\scriptscriptstyle goal}^{{\rm\scriptscriptstyle FEASIBLE}}) \geq 0$ by optimality, we obtain that $J_1 > J_2$. This shows that the second (non-steady state) solution, associated to the cost $J_2$, is more convenient than the one associated to $J_1$, and so $z_{\rm\scriptscriptstyle ref}^{\infty}= z_{\rm\scriptscriptstyle goal}^{{\rm\scriptscriptstyle FEASIBLE}}$ is the only possible steady-state solution where the output $\bar{z}_0(k|k)$ converges, contradicting the assumption. This entails that, as $k\rightarrow \infty$, $\bar{z}(k|k) \to z_{\rm\scriptscriptstyle goal}^{{\rm\scriptscriptstyle FEASIBLE}}$.\\
Finally, in view of \eqref{eq:constrZ} we obtain that  $\text{dist}(z(k),z_{\rm\scriptscriptstyle goal}^{{\rm\scriptscriptstyle FEASIBLE}}\oplus C(\bar{\mathbb{E}}\oplus\hat{\mathbb{E}})) \to 0$ as $k \to \infty$.

%\bibliography{mybiblio}
%\bibliographystyle{plain}

\end{document}